\documentclass[11pt,letter]{article}
\usepackage{amsfonts,amssymb,amsmath,setspace,graphicx,amsthm}
\usepackage{mathrsfs}
\usepackage{color}
\usepackage[margin=1in]{geometry}  % ITCS submission requires margins of at least 1in
\usepackage{graphicx}
\usepackage[style=alphabetic,sorting=nyt,maxnames=4,maxcitenames=4,maxbibnames=4,backend=bibtex]{biblatex}
\usepackage[usenames,dvipsnames,svgnames,table]{xcolor}
\usepackage{framed}
\usepackage[T1]{fontenc}
\usepackage{tikz}
\usetikzlibrary{arrows}
\usetikzlibrary{positioning}
\usepackage{tkz-graph}
\usepackage{authblk}
\usepackage{wrapfig}

\usepackage{bold-extra}
\usepackage{algorithm,algorithmic}

\newenvironment{idealfunc}[1]
  {\noindent\rule{\columnwidth}{0.8pt}
    {\bf Ideal functionality} #1

    \noindent\rule{\columnwidth}{0.4pt}
    \small\it
  }
  {\noindent\rule{\columnwidth}{0.8pt}
  }

\newcounter{protocol}
\newenvironment{protocol}[1]
  {\refstepcounter{protocol}
    \noindent\rule{\columnwidth}{0.8pt} \nopagebreak
    {\bf Protocol \theprotocol.} #1

    \nopagebreak 
    \noindent\rule{\columnwidth}{0.4pt}
    \small\it
    \nopagebreak\vskip0.1pt\nopagebreak
  }
  {\nopagebreak\vskip0.1pt\nopagebreak
  \noindent\rule{\columnwidth}{0.8pt}
  }

\newenvironment{informalthm}
  {\medskip\noindent{\bf Theorem (informal).}}
  {\medskip}

\newcommand{\defeq}{\stackrel{\mathrm{def}}{=}}

\newtheorem{theorem}{Theorem}[section]
\newtheorem{property}[theorem]{Property}

\newtheorem{lemma}[theorem]{Lemma}
\newtheorem{definition}[theorem]{Definition}

\newtheorem{assumption}{Assumption}

\newcommand{\score}{s}
\newcommand{\pparams}{\mathsf{CK}}

\newcommand{\xor}{\oplus}

\newcommand{\Share}{{\sf Share}}
\newcommand{\Recon}{{\sf Reconstruct}}
\newcommand{\outp}{{\sf out}}

\newcommand{\powset}{\mathcal{P}}

\newcommand{\zo}{ \{0,1\}}

\newcommand{\cC}{{\cal C}}
\newcommand{\cD}{{\cal D}}

\newcommand{\cH}{{\cal H}}
\newcommand{\cI}{{\cal I}}

\newcommand{\cS}{{\cal S}}
\newcommand{\cT}{{\cal T}}

\newcommand{\cX}{{\cal X}}
\newcommand{\cY}{{\cal Y}}
\newcommand{\cZ}{{\cal Z}}

\newcommand{\sM}{\mathscr{M}}

\newcommand{\sO}{\mathscr{O}}

\newcommand{\bE}{\mathbb{E}}

\newcommand{\bI}{\mathbb{I}}

\newcommand{\bR}{\mathbb{R}}

\newcommand{\from}{ \leftarrow}

\newcommand{\NL}[1]{\tabularnewline~&~&~\tabularnewline }

\newcommand{\eps}{\varepsilon}
\newcommand{\rarr}{\rightarrow}
\newcommand{\larr}{\leftarrow}
\newcommand{\RR}{\mathbb{R}}
\newcommand{\NN}{\mathbb{N}}
\newcommand{\ZZ}{\mathbb{Z}}
\newcommand{\poly}{\mathsf{poly}}
\newcommand{\negl}{\mathsf{negl}}

\newcommand{\Gen}{\mathsf{Gen}}
\newcommand{\Enc}{\mathsf{Enc}}

\newcommand{\Sign}{{\sf Sign}}
\newcommand{\Verify}{{\sf Verify}}
\newcommand{\MAC}{{\sf MAC}}

\newcommand{\compIndist}{\overset{c}{\approx}}

\newcommand{\ack}{{\sf ack}}
\newcommand{\quit}{{\sf quit}}
\newcommand{\chall}{{\sf challenge}}
\newcommand{\resp}{{\sf response}}

\newcommand{\Adv}{{\cal A}}
\newcommand{\Sim}{{\cal S}}

\newcommand{\view}{\mathscr{V}}
\newcommand{\rview}{V}

\newcommand{\before}{{\rm before}}
\newcommand{\after}{{\rm after}}

\newcommand{\HidingExp}{{\sf HidingExp}}
\newcommand{\StrongHidingExp}{{\sf StrongHidingExp}}

\newcommand{\IdealSup}{{\sf ideal}}
\newcommand{\RealSup}{{\sf real}}
\newcommand{\IdealFuncOrdered}{{\cal F}_{\rm Ordered\mhyphen MPC}}

\newcommand{\IdealFuncQueued}{{\cal F}_{\rm Timed\mhyphen Delay\mhyphen MPC}}
\newcommand{\IdealFuncMPC}{{\cal F}_{\rm MPC}}
\newcommand{\IdealViewOrdered}{{\view}^{\IdealSup}_{\rm Ordered\mhyphen MPC}}
\newcommand{\IdealViewQueued}{{\view}^{\IdealSup}_{\rm Timed\mhyphen Delay\mhyphen MPC}}
\newcommand{\IdealViewMPC}{{\view}^{\IdealSup}_{\rm MPC}}
\newcommand{\RealViewQueued}{{\rview}^{\RealSup}}
\newcommand{\RealViewOrdered}{{\rview}^{\RealSup}}

\newcommand{\Lock}{{\sf Lock}}
\newcommand{\Unlock}{{\sf Unlock}}
\newcommand{\TimeStep}{{\sf TimeStep}}
\newcommand{\IterateTimeStep}{{\sf IterateTimeStep}}
\newcommand{\UnlockFull}{{\sf CompleteUnlock}}

\newcommand{\PPT}{{\sc ppt}}
\newcommand{\hatH}{\widetilde{\cH}}
\newcommand{\hath}{\widetilde{h}}
\newcommand{\RSA}{\mathsf{RSA}}
\renewcommand{\time}{{\sf time}}

\newcommand{\SeqTimeLine}{\mathsf{BB\mhyphen TimeLinePuzzle}}
\newcommand{\RSWTimeLine}{\mathsf{Square\mhyphen TimeLinePuzzle}}

\newcommand{\rankOne}{\textrm{$n$-dimensional}}
\newcommand{\rankTwo}{\textrm{$n^2$-dimensional}}
\newcommand{\ShareData}{\mathsf{SHARE\mhyphen DATA}}
\newcommand{\MinWeightFAS}{\sf Minimum\mhyphen Weight\mhyphen Feedback\mhyphen Arc\mhyphen Set}

\mathchardef\mhyphen="2D
\DeclareMathOperator*{\argmin}{arg\,min}

\theoremstyle{definition}
\newtheorem{remark}{Remark}

\newlength{\saveparindent}
\newlength{\saveparskip}
\newcommand{\NP}{\text{NP}}

\newif\ifproportionalfairness\proportionalfairnessfalse
\newif\ifhypotheses\hypothesestrue
\newif\ifdoublecolumn\doublecolumnfalse

\begin{document}

\pagenumbering{gobble}
\title{How to Incentivize Data-Driven Collaboration \\ Among Competing Parties}
%\title{On Fairness and Timing in \\ Multiparty Computation and Collaboration}

\author[1]{Pablo Daniel Azar\thanks{Supported by the Robert Solow Fellowship 3310100.}}
\author[2]{Shafi Goldwasser\thanks{Supported by NSF Eager CNS-1347364, NSF Frontier CNS-1413920, the Simons Foundation (agreement dated June 5, 2012), Air Force Laboratory FA8750-11-2-0225, and Lincoln Lab PO7000261954. This work was done in part while these authors were visiting the Simons Institute for the Theory of Computing, supported by the Simons Foundation and by the DIMACS/Simons Collaboration in Cryptography through NSF grant CNS-1523467.}}
\author[1]{Sunoo Park\protect\footnotemark[2]}
\affil[1]{MIT}
\affil[2]{MIT and the Weizmann Institute of Science}

%\author{Pablo Azar\inst{1} \and Shafi Goldwasser\inst{2} \and Sunoo Park\inst{1}}
%\institute{MIT \and MIT and the Weizmann Institute of Science}

%\date{{\tt \{azar,shafi,sunoo\}@csail.mit.edu} \\ \bigskip MIT}
\date{}
\maketitle

\thispagestyle{empty}

\begin{abstract}
The availability of vast amounts of data is changing how we can make medical discoveries, 
predict global market trends, save energy, 
%improve our infrastructures, 
and develop new educational strategies. 
In certain settings such as Genome Wide Association Studies or deep learning,
the sheer size of data (patient files or  labeled examples) seems critical to making discoveries. 
When data is held distributedly by many parties, as often is the case, they must share it to reap its full benefits.

One obstacle to this revolution is the lack of willingness of different entities to share their data, due
to reasons such as possible loss of privacy or competitive edge.
Whereas cryptographic works address the privacy aspects, 
they shed no light on individual parties' losses and gains when access to data carries tangible rewards.
Even if it is clear that better overall conclusions can be drawn fom collaboration, are individual collaborators better off
by collaborating? Addressing this question is the topic of this paper.

Our contributions are as follows. 
\begin{itemize}
\item We formalize a model of $n$-party collaboration for computing functions over private inputs in which the participants receive their outputs in sequence, and the order depends on their private inputs. Each output ``improves'' on all previous outputs according to a score function. 
\item We say that a mechanism for collaboration achieves a \emph{collaborative equilibrium} if it guarantees a higher reward for all participants when joining a collaboration compared to not joining it. We show that while in general computing a collaborative equilibrium is NP-complete, we can design polynomial-time algorithms for computing it for a range of natural model settings. When possible, we  design mechanisms to compute a distribution of outputs and an ordering of output delivery, based on the $n$ participants' private inputs, which achieves a collaborative equilibrium.
\end{itemize}

The collaboration mechanisms we develop are in the standard model, and thus require a central trusted party; however, we show that this assumption is not necessary under standard cryptographic assumptions. We show how the mechanisms can be implemented in a decentralized way by $n$ distrustful parties using new extensions of classical secure multiparty computation that impose order and timing constraints on the delivery of outputs to different players, in addition to guaranteeing privacy and correctness.

\end{abstract}

\newpage
\setcounter{page}{1}

%\newpage
%\tableofcontents

%\newpage
\pagenumbering{arabic} 
\section{Introduction}

The availability of vast amounts of data is changing how we can make medical discoveries, 
predict global market trends, save energy, improve our infrastructures, and develop new educational strategies. 
Indeed,  it is becoming clearer that \emph{sample size} may be the most important factor in 
making surprising new discoveries in a number of areas such as \emph{genome-wide association studies}\footnote{A genome-wide association study
is an investigation of common genetic variants in a population, in order to identify genetic variants that are associated with a given trait.} 
(GWAS) and \emph{machine learning} (ML),
as witnessed by the striking success of GWAS studies with large samples for 
schizophrenia\footnote{``{\sf Dramatic increase in patient data
size enabled the discovery of more than 100 gene loci associated with the disease up from a handful loci seen with small sets of
patients. This was made possible due to an unusually large scale collaborations among many institutes.}''} 
\cite{schizo1,schizo2,schizo3}
and the success of deep learning in ML.

When large data is required, parts of the data are often held by different entities.
Such entities need to share their data, or at least engage in a collaborative 
computation where each entity manages its own private data, in order for society to reap the benefit of large sample sizes. Referring back to the 
GWAS example, success was explicitly attributed to such collaboration: ``{\sf The schizophrenia  study was made possible 
due to unusually large scale collaborations among many institutes... This level of cooperation between institutions is absolutely essential... 
If we are to continue elucidating the biology of psychiatric disease through genomic research, we must continue to work together.}'' 
\cite{schizophreniaQuote}

Unfortunately, the above example is the exception rather than the rule. A major obstacle to the big-data revolution is the lack
of willingness of different entities to share data in collaborations with each other: so-called  ``data hoarding''.
One obstacle is privacy concerns, where parties refuse to collaborate, in order to protect the privacy of their data.
%Computation while maintaining privacy  has been addressed extensively in the last 30 years \cite{BGW88, GMW87,FHE,FE} and is not the subject of this paper.
Privacy, however, is not the only obstacle.

An equally important obstacle is competition between entities holding data.  
When access to data carries tangible rewards, say, if the entities are 
companies competing for a share of the same market or research laboratories competing for scientific credit, 
it is unclear whether an individual collaborator
is better off, even if it is clear that better overall conclusions can be drawn from collaboration. 
Stated in more game-theoretic terms,
the entities face the following dilemma:
{\it whereas the overall societal benefit of collaboration is clear, the utility for an individual collaborator may be negative,
so why collaborate? }  
Addressing this question is the  topic of this paper. 

In this paper, we present a formal model for collaboration in which this question can be analyzed, as well as
design mechanisms to enable collaboration where all collaborators 
are provably ``better off'', when possible.
The \emph{order} in which collaborators receive the outputs of a collaboration will be a crucial aspect of our model and mechanisms.
We believe that timing is an important and primarily unaddressed issue in data-based collaborations. 
For example, in the scientific research community, data sharing can translate to losing a prior publication date.
In financial enterprises, the timing of investments and stock trading can translate to large financial gains or losses.

\iffalse
The collaboration mechanisms we develop assume a  central trusted party; 
however, we will show that this assumption is not necessary under standard cryptographic assumptions. 
We show in Section \ref{sec:queuedDef} how the mechanisms can be implemented in a decentralized way 
by $n$ distrustful parties by using new extensions of classical secure multiparty computation
that give formal guarantees on the \emph{order and timing} of output delivery as well as the standard requirements of privacy and correctness.
\fi

We show in Section \ref{sec:queuedDef} that the collaboration mechanisms we develop can be implemented in a decentralized way by $n$ distrustful parties even in the presence of a subset of colluding polynomial-time parties who may deviate in an arbitrary fashion, under standard cryptographic assumptions. To achieve this, we extend the theory of multi-party computation (MPC) to impose order and time on the delivery of outputs to different players.

\subsection{Summary of our contributions}

\subsubsection{A model of collaboration}\label{introsec:model}

We propose a model for collaboration
which enables the determination of whether the utility obtained by a collaborator outweighs the utility he may obtain without collaboration. 
The ultimate desired outcome of a collaboration is to learn a parameter of the (unknown) joint distribution
from which the participants' input data $x_1,\dots,x_n$ is drawn.
This can be expressed as $y^*=f(\cX)$ where $\cX$ is the joint distribution of input data and $f$ is a known function.
In our model, the outcome of a collaboration is a pair $(\pi,\vec{\cZ})$ where $\pi$ is a permutation
of player identities and $\vec{\cZ}=(\cZ_1,\dots,\cZ_n)$ where each $\cZ_{\pi(i)}$ is a distribution that corresponds to
player $i$'s ``estimate'' of $y^*$. We think of $\cZ_{\pi(i)}$ as the public output of player $i$:
for example, in the setting of scientific collaboration, $\cZ_{\pi(i)}$ would be player $i$'s academic publication.
Our model setup assumes an underlying score function which assigns scores to the players' outputs.
%The quality of a player's output $\cZ_{\pi(i)}$ will be measured by a score function $s$

\iffalse
The ultimate desired outcome of a collaboration  
is to learn the output $y^*$ of a function $f$ evaluated on participants' input data $x_1,\dots,x_n$.
In our model, the outcome of a collaboration is a pair $(\pi,\mathbf{z})$ where  $\pi$ is a permutation of
player identities  and $\mathbf{z}=(z_1,,\dots,z_n)$ is a vector of ``public outputs'' where $z_{\pi(i)}$ is the output of player $i$.
For example, in the setting of scientific collaboration, the public output could be an academic publication.
Our model setup assumes an underlying distance metric $d$ in which any two outputs can be compared.
The quality of a player's output $z_i$ will be higher when the distance $d(z_i,y^*)$ to the ``true output'' $y^*$ is lower.
\fi

The model includes a \emph{reward function} $R_t$ which characterizes the gain in utility for any given party $i$ in a collaboration.
The reward that a party $i$ gets depends on how much his score $\score(\cZ_{\pi(i)})$ \emph{improves on} the previous state of the art $\score(\cZ_{\pi(i)-1})$, 
and on $\pi(i)$, namely, \emph{when} the party makes his public output.
Specifically, the reward function includes a multiplicative \emph{discount factor} $\beta^t$ where $\beta\in[0,1]$ and $t$ is the time of publication,
meaning that the reward from a publication is ``discounted'' more as time goes on.
$$R_t(\pi,\vec{\cZ}) = \beta^{t} \cdot (\score(\cZ_{\pi(t-1)})-\score(\cZ_{\pi(t)}))$$

To determine whether the utility of collaboration outweighs the utility of working on one's own,
our model uses ``outside payoff'' values $\alpha_i$ which are the score that party $i$ would obtain \emph{without collaborating}.
%\footnote{Our model is flexible enough that we could have $\alpha_i' = \alpha_i  + c_i$ where $\alpha_i$ is the utility to player $i$ from not collaborating and $c_i$ is some overhead cost of collaboration. Alternatively, we could have $\alpha_i' = \alpha_i \cdot c_i$, so that agent $i$ will collaborate only if her reward from doing so is greater than $c_i$ times her reward for not collaborating. (Indeed, $\alpha_i'$ could be any function of $\alpha_i$.) All our results would follow through by replacing $\alpha_i$ with the corresponding $\alpha_i'$.}
%Our model works with ``outside payoff'' values $\alpha_i$ which indicate the utility that party $i$ would obtain \emph{without collaborating}:
%These values may be provided by the relevant party $i$ or, in some scenarios (as discussed in the examples given in the next subsection), 
$\alpha_i$ can be computed directly from the input $x_i$ of party $i$. 

\subsubsection{Mechanisms and collaborative equilibrium}\label{introsec:mech}
%\noindent{\bf Mechanisms for collaboration and collaborative equilibrium} 
We define a notion of \emph{collaborative equilibrium} in which all parties are guaranteed a non-negative reward,
and develop {\em mechanisms} for collaboration that compute such equilibria.  
When an equilibrium exists, our mechanism
delivers a sequence of progressively improving ``partial information'' about $y^*$ to the collaborating parties.
More specifically, the mechanism will take as input the data of all parties, and output
a pair $(\pi,\vec{\cY})$ where  $\pi$ is a permutation of
player identities  and $\vec{\cY}=(\cY_1,,\dots,\cY_n)$ specifies the outcomes to be delivered to the players: 
each $\cY_{\pi(i)}$ is the approximation to $y^*$ that is given to player $i$
at time-step $\pi(i)$, such that the score of the outputs is increasing with time. That is, $\score(\cY_{\pi(1)})>\dots>\score(\cY_{\pi(n)})$.
We emphasize that both the order $\pi$ and the outputs $\cY_i$ are computed based on the inputs of all players.

%One of the main features of our model is that, even though 
When player $i$ receives an output $\cY_{\pi(i)}$ from the central mechanism, 
she may combine $\cY_{\pi(i)}$ with the information that she learned from prior public outputs and her own input $x_i$,
to generate a public output $\cZ_{\pi(i)}$.
We first prove that the ability of the players to learn from others' publications, in general, will make
the problem of deciding whether there exists an equilibrium is $\NP$-complete (see Theorem \ref{thm:NPC}).

Next, we show that there is a polynomial-time mechanism that can output an equilibrium whenever one exists (or output {\tt NONE} if one does not exist) 
for a variety of model settings and parameters which we characterize (see Theorem \ref{thm:mechanismExists}). 
An example of a setting when a polynomial-time mechanism is possible is when
\begin{itemize}
\item  there is an upper bound $\mu_j$ on the amount of information that any player can learn from a given player $j$'s publication, and 
\item it is possible to efficiently compute, for any $y^*$ and $\delta>0$, an ``approximation'' $\cY'$ such that $\score(\cY')=\delta$.
\end{itemize}
%Secondly, the possibility of a polynomial-time mechanism will depend on how we bound the amount of information the players can learn from
%others' prior publications.
%\iffalse

%In our model, each party $i$ may produce a public output $z_{\pi(i)}$ upon receiving his result $y_{\pi(i)}$ from the collaboration.
%This may, for example, be an academic publication ($z_{\pi(i)}$) based on the information ($y_{\pi(i)}$) obtained by collaborating.
\iffalse
To state our theorems, we will use the terminology of scientific collaboration and publications (though we emphasize that
our results apply to more general settings).
A delicate issue with sequential publications is that players may learn from others' prior publications
and incorporate this extra information into their eventual $z_{\pi(i)}$. In order to give meaningful results,
we assume upper bounds on the amount of information that players can learn through these side channels. The way we bound this information
will make a big difference to whether sharing data in our model is computationally efficient or intractable.
\fi

%Let $y^*=f(x_1,\dots,x_n)$ be the ``ideal'' outcome of a collaboration.
%Informally, if a data-sharing mechanism gives player $i$ a result $y_{\pi(i)}$, then player $i$ may publish any result $z_{\pi(i)}$ such that $d(z_{\pi(i)},y^*)\geq d(y_{\pi(i)},y^*)-\lambda$, where $\lambda$ is a {\em learning bound} that may depend on previous players' publications $z_{\pi(j)}$ as well as the order in which these results were revealed.
In a nutshell,
\iffalse 
In particular, the mechanism uses the bounds when the amount that any player can learn from a given player $j$'s publication is bounded by the same $\mu_j$,
then there exists an efficient data-sharing mechanism.
\begin{informalthm}
Assume that when player $j$ publishes, all players $i$ who have not published yet can reduce the variance of their own estimate by at most $\mu_j$. Then \emph{there exists a polynomial-time mechanism} that outputs an ordering of players $\pi$ and approximations $(y_1,\dots,y_n)$ specifying a collaborative equilibrium. If no such ordering and estimates exist, the mechanism outputs ${\tt NONE}$.
\end{informalthm}
\fi
the bounds $\mu_j$ can be used to define a weighted graph in which the weight of the minimum-weight perfect matching 
determines the existence of a collaborative equilibrium.

\subsubsection{Cryptographic protocols to implement the mechanisms}\label{introsec:crypto}
%\noindent{\bf Cryptographic Protocols implementing the Mechaanisms}
We develop cryptographic protocols for implementing the mechanisms without a centralized trusted party
and in the presence of a subset of colluding players who may deviate from the protocol in an arbitrary fashion, under cryptographic assumptions.  
The protocols compute the collaboration outcome $(\pi,\vec{\cY})$ via multi-party secure computation on players' private inputs. 
Since a crucial aspect of the mechanism's ability to yield non-negative reward to all players
is the delivery of outputs in order,  we need to extend the classical notion of  MPC  to incorporate
guarantees on the order and timing of output delivery.
These extensions may be of interest independent 
of the application of mechanisms for incentivizing collaborations.

We define \emph{ordered MPC} as follows.
Let $f$ be an arbitrary $n$-ary function and $p$ be an $n$-ary function that outputs permutation $[n]\rarr[n]$.
An ordered MPC protocol is executed by $n$ parties, where each party $i\in[n]$ has a private input $x_i\in\{0,1\}^*$,
who wish to securely compute
$f(x_1,\dots,x_n) = (y_1,\dots,y_n)$
where $y_i$ is the output of party $i$.
Moreover, the parties are to receive their outputs in a particular \emph{ordering} dictated by
$p(x_1,\dots,x_n) = \pi$ where $\pi$ is a permutation of the player identities.
\iffalse
That is, for all $i<j$, party $\pi(i)$ must receive his output \emph{before} party $\pi(j)$ receives her output.

In the latter,  
the outputs are delivered to the $n$ players 
according to a particular permutation $\pi:[n]\rarr[n]$,
which is computed during the MPC and \emph{may itself be a function of the players' inputs}.
\fi
%Note that the output ordering $\pi$ is data-dependent, as $p$ is a function of the parties' inputs.
Since the choice of $\pi$ depends on private inputs, it may leak information: hence,
%Since $\pi$ depends on private data, it may leak information: hence,
we formulate an enhanced \emph{privacy} requirement for ordered MPC
that each player should learn his output 
and his \emph{own} position in the output ordering, and nothing more (see Definition \ref{def:orderedSecurity}).

%The traditional MPC \emph{fairness} requirement is, informally, that either all parties learn their outputs, or none do.
%The thesis of our paper is that \emph{fairness is not enough}. \spnote{should change ``thesis'' to some other word?}
%In certain settings, the ordering of when outputs are delivered may be just as important. 
%We thus refine the classical notions of guaranteed output delivery and fairness to require instead {\it ordered output delivery} and {\it prefix-fairness}.

%A well-known result of \cite{Cleve86} is that fairness is impossible to achieve for general functionalities when a majority of players are dishonest.
%We refine the classical notion of MPC fairness to require instead 
%{\em prefix-fairness} which guarantees that if a pair of players $i,j$ receive outputs that will receive the outputs according to the mandated offer. 

We show a simple transformation from classical MPC protocols for general functionalities $f$ to 
ordered MPC protocols for general functionalities $f$ and permutation functions $p$ that achieve enhanced privacy,
even when a minority of the $n$ players may be colluding to sabotage the protocol (see Theorem \ref{thm:orderedProtocol}).
The assumptions necessary are the same as for the classical MPC constructions (e.g. \cite{GMW87}).
When the colluding players are in majority, it is well known that output delivery to all honest parties cannot be guaranteed \cite{Cleve86}.

\iffalse
When the faulty players are in minority\footnote{\label{ft:fairness}Our ``ordered output delivery'' is a stronger notion than traditional MPC fairness.
Since achieving fairness is known to be impossible in general when a majority of parties is faulty, 
we cannot hope to achieve ordered output delivery in the dishonest majority setting. We do discuss notions of fairness and propose a new
notion called prefix-fairness which is achievable for any number of faulty players; details are given in later sections.}, 
we show simple extensions to classical MPC constructions yield new protocols 
where all players are guaranteed to receive their outputs in the mandated order.
Our ``ordered output delivery'' is a stronger notion than traditional MPC fairness.
Since achieving fairness is known to be impossible in general when a majority of parties is faulty, 
we cannot hope to achieve ordered output delivery in the dishonest majority setting.
%When the faulty players are in majority, the protocols will still be guaranteed to achieve prefix-fairness.
The assumptions necessary are the same as for the classical MPC constructions (e.g. \cite{GMW87}).
\fi

Next, we define {\em timed-delay MPC}, where explicit time delays are introduced into the output delivery schedule. 
Time delays between the outputs may be crucial to enable parties to reap the benefits of their position in the order.
We give two constructions of timed-delay MPC in the honest majority setting\footnote{We cannot hope to achieve timed-delay MPC
in the case of dishonest majority since, as mentioned in the preceding paragraph, even output delivery cannot be guaranteed in this setting.}. 
First, we give a conceptually simple protocol which runs ``dummy rounds'' of communication
in between issuing outputs to different players, in order to measure time-delays. The simple protocol has the flaw that
all (honest) players must continue to interact until the last party receives his output 
(that is, they must stay online until all the time-delays have elapsed).
To address this issue, we present a second protocol assuming the existence of time-lock puzzles \cite{RSW96} in addition to
the classical MPC \cite{GMW87} assumptions (see Theorem \ref{thm:timeLockProtocol}).
Informally, a time-lock puzzle is a primitive which allows ``locking'' of data, such that it will be released after a pre-specified
time delay, and no earlier. Our second timed-delay MPC protocol, instead of issuing outputs to players in the clear, gives to each party
his output \emph{locked} into a time-lock puzzle; and in order to enforce the desired ordering, the delays required to unlock the puzzles
are set to be an increasing sequence.
An issue that arises when giving out time-lock puzzles to many parties is that different parties may have different computing power, and hence
solve their puzzles at different speeds: for example, it is clear that we cannot guarantee that players learn their outputs in the
desired ordering if some players compute arbitrarily faster than others. 
Still, we show that
our protocol is secure and achieves ordered output delivery
in the case that the difference between any two players' computing power is known to be bounded by a logarithmic factor. If the
assumption about computing power does not hold, then the protocol still achieves security (i.e. correctness and privacy), 
but the ordering of outputs is not guaranteed.

The definition of ordered and timed-delay MPC inspire new notions unrelated to the central topic of this paper.
In particular:

\begin{itemize}
\item {\bf Time-lines.}
%\noindent{\bf Time-lines}
Inspired by the application of time-lock puzzles to time-delayed MPC,
we propose the new concept of a  {\em time-line}, where multiple data items can be locked so that
their unlocking must be serialized in (future) time. See Section \ref{sec:timeLine} for details.

\item {\bf Prefix-fairness.}
In the traditional MPC landscape, fairness is the one notion that addresses the idea that either all parties participating in an 
MPC should benefit, or none should. Fairness requires that either all players receive their output, or none do.
However, it is well-known that fairness is achievable when a majority of the players are honest, but it is \emph{not} achievable for general functionalities when a majority of players are faulty \cite{Cleve86}.
We propose a refinement of the classical notion of fairness in the setting of ordered MPC, called \emph{prefix-fairness},
where players are to receive their outputs one after the other according to a given ordering $\pi$, and the guarantee is that
\emph{either} no players receive an output \emph{or} those who do strictly belong to 
a prefix of the mandated order $\pi$ (see Definition \ref{def:fairOrdered}).
Prefix-fairness can be achieved for general functionalities and \emph{any number} of faulty players, 
under the same assumptions as classical MPC \cite{GMW87} (see Theorem \ref{thm:prefixFairness}).
\end{itemize}

\subsection{Discussion and interpretation of our work}

\paragraph{Slowing down scientific discovery?}
%\noindent{\bf Slowing Down Scientific Discovery?}
Intuitively, the mechanisms we develop  always take the following form: 
the mechanism computes the ``best possible estimate'' $\cY^*$ of $y^*$ given the input data of the players,
and then hands out a sequence of successively more accurate
(according to the score function) outcomes, where the final party receives $\cY^*$.

%One objection may take to the format of our mechanism, is that 
%One perspective on our mechanism format is the following:
%even though a better collaboration outcome can be computed based on the entirety of all players' inputs, 
%the mechanism hands to collaborators worse outcomes for a while, slowing down the process of discovery.
One may ask: why slow down scientific progress and hand out inferior results when better ones are available?
We argue that progress will in fact be \emph{enhanced}, not slowed down, by this methodology,
as it will be a decisive factor in parties' willingness to collaborate in the first place.
This bears great similarity to the original philosophy of \emph{differential privacy} and privacy-preserving
data analysis more generally. In these fields, accuracy (so-called utility) of answers to aggregate queries over items in database is partially sacrificed
in order to preserve privacy of individual data items, as a way to encourage individuals to contribute their data items to 
the database. In an analogous way, in order to get results based on the large data sets held by potential collaborators, we sacrifice
the \emph{speed} of discovery of the ``ultimate'' collaboration outcome: we are willing to pay this price to incentivize
parties to collaborate and contribute their data. In contrast to differential privacy,
we do not sacrifice ultimate accuracy. 
The last collaborator to receive an output, receives the ideal outcome $\cY^*$.  Namely, $\cY_{n}=\cY^*$. 

%A few remarks are in order.

\paragraph{The Fort Lauderdale example: the importance of time.}
A recurring idea in this work is the importance of time and ordering of research discoveries, which is inspired in part by the
following striking example from the field of genomics.
%For example, in the scientific research community, data sharing can translate to losing a prior publication date. 
%This issue explicitly came up at
%An interesting real-life case is 
In the 2003 Fort Lauderdale meeting on large-scale biological research \cite{FortLauderdale}, 
the gathering of leading researchers in the field recognized that 
``{\sf pre-publication data release can promote the best interests of [the field of genomics]}''
but ``{\sf might conflict with a fundamental scientific incentive -- publishing the first analysis of one's own data}''.
Researchers at the meeting agreed to adopt a set of principles by which although data is shared upon discovery, 
 researchers hold off publication until the original holder of the data has published a first analysis. 
Being a close-knit community in which reputation is key, this was a viable agreement which has led to great productivity and advancement of the field.
However, more generally, their report states that ``{\sf incentives should be developed by the scientific community to support the voluntary release of
[all sorts of] data prior to publication}''.
This example teaches us to focus on three key aspects of collaboration:
the incentive to collaborate has to be clear to all collaborators;
there must be a way to ensure adherence to the rules of collaboration;
and timing is of the essence.
%\spnote{can we phrase this last sentence better?}
%A central motivation of our present work is the development of rigorous approaches towards this goal.
%This  inspired us to focus on the order of output delivery as a crucial aspect of our data-sharing mechanisms,
%since we believe that this is an important and primarily unaddressed issue in data-based collaborations.
%Namely, although better conclusions can be drawn from a combined data set, pooling data may not be worth the gain for an individual participant, 
%if it implies enabling an earlier discovery by the competition.
 
\paragraph{Privacy implies increased utility.}
%\noindent{\bf Privacy of Individual Collaborator Data :} 
\iffalse
As we remarked above,  parties may refuse to collaborate because of privacy concerns or legal regulations.
This has been well-studied over the the last 30 years. In particular, the theory of multi-party computation (MPC) \cite{GMW87,BGW88}
enables $n$ parties to compute a function on the union of their data, revealing nothing but the output of
the function even in presence of colluding, malicious coalition of players. 
\fi
Although the goal of our work is to design mechanisms {\em to incentivize collaboration}
by increasing the utility of collaborations rather than focusing on the 
privacy of individual entities' input data, MPC protocols prove to be an important technical tool to implement the mechanisms which
guarantee increased utility. As a by-product, the use of MPC 
provides our mechanisms with the additional guarantee of privacy.
\iffalse
Another area of research where privacy is the central concern is differential privacy,
where the challenge is to provide statistical answers to aggregate queries computed over
individual data items held in a database, in a way that
ensures that the removal or addition of any single database item does not
(substantially) affect the outcome of any analysis.
Our work does not fall within the differential privacy framework, but shares a philosophy with differential privacy
as discussed below.
\fi

\paragraph{Future directions}
When collaboration is feasible, each party $i$ in our model is guaranteed a  reward from collaborating that is greater than the reward $\alpha_i$ they could get on their own.  
%However, there are many settings where ``the sum is greater than the parts'', and the contribution that player $i$'s data $x_i$ makes to $\cY^*$ (i.e. the best possible estimate of $y^*$ given the joint input data) is much larger than the value $\alpha_i$ of just knowing $x_i$ in isolation. 
However, the contributions of the players' data to the computation of the final output $\cY^*$ 
%(i.e. the best possible estimate of $y^*$ given the joint input data) 
may be asymmetric: some special player $i^*$ may have some data that helps solve the ``puzzle'', but this player $i^*$ may not be known a priori before the participants decide to collaborate\footnote{An example in the same vein is the following. In the medical setting, a hospital with a larger patient population will clearly have more patient data than a small facility,
and yet access to data of small but homogeneous or rare communities can at times be more valuable than access to larger heterogeneous sets of data.} An interesting future direction would be developing mechanisms where, even without a priori knowledge of which players have higher quality data, 
we can still design collaborations where the players whose contribution turned out most valuable get most credit.
%we can still allocate more credit  after the collaboration is done to the players whose contribution turned out the most valuable.

Another future direction of interest to design \emph{truthful} mechanisms so that collaborating parties will be provably
incentivized to submit their true and accurate data as input. In our work, we assume that, while we can incentivize the players to collaborate or not, once they decide to collaborate they are truthful about the value of their dataset $x_i$. From the point of view of scientific publications, this assumption is reasonable if we believe that the experiments that generate this data can be verified or replicated, and that a failure to replicate would hurt a scientific group's reputation. However, there are many settings, such as businesses pooling their data together to generate larger profits, where the parties may be incentivized to lie about their output $x_i$. Since we are already assuming that parties are rational, a future direction would be to develop mechanisms where, even when parties can lie about $x_i$ (because $x_i$ cannot be verified by others), they are still incentivized to report it truthfully. One possible direction is where $x_i$ is the output of some long $\#P$ computation (for example, a Markov Chain Monte-Carlo simulation), where (a) replicating the computation would take a very long time and delay publication for everyone in the group and (b) player $i$ cannot prove in a classical way that their output $x_i$ is correct. Even in this case,  player $i$ can be incentivized to give the right answer via a rational proof \cite{RationalProofs,SuperEfficientRationalProofs,Hubacek}. 

\iffalse
Our setting is useful and most likely to lead to collaboration when there are increasing marginal
returns from adding new data. It will be interesting to discover new settings where this is provably the case.
Interestingly, this does not include all learning problems whose objectives can be stated as minimizing a convex loss 
function (or maximizing a concave value function). Indeed, 
solving new machine learning problems with non-convex objectives have recently become practical. We refer to
Bengio and LeCun  \cite{LeCun} for a more thorough discussion of this situation. 
\fi

Our setting is useful and most likely to lead to collaboration when there are increasing marginal
returns from adding new data. It will be interesting to discover new settings where this is provably the case.

\subsection{Other related work}

The problem of how to make progress in a scientific community has been studied in other contexts. 
Banerjee, Goel and Krishnaswamy \cite{BanerjeeEtAl} consider the problem of partial progress sharing, where a scientific task is modeled as a directed acyclic graph of subtasks. 
Their goal is to minimize the time for all tasks to be completed by selfish agents who may not wish to share partial progress.
%A task can be undertaken if and only if all the tasks that precede it in the topological order of the network have been completed. 
%The goal of a social planner is to minimize the time it takes to complete all the tasks in the network, subject to the above ordering constraint. 
%However, the tasks are not performed by a social planner. Instead, they are performed by selfish agents, each of whom has different abilities, 
%and who can chose which subtasks to work on (as long as they already know the results of previous subtasks on the topological order of the network) 
%and when to share with other players the results of the subtasks they work on. The authors show sufficient conditions for immediately sharing all progress to be a Bayes Nash Equilibrium, and show that, %from a welfare point of view, there exist cases where not sharing progress immediately can increase the minimum makespan (completion time of the last task), so it is important to ensure good mechanisms for sharing progress. 

Kleinberg and Oren \cite{KleinbergOren} study a model where researchers have different projects to choose from, and can work on at most one. 
Each researcher $i$ has a certain probability of being able to solve a problem $j$, and she gets a reward $w_j$ if she is the only person to solve it. If multiple researchers solve the problem, they study how to split the reward in a socially optimal way. 
They show that assigning credit asymmetrically can be socially optimal when researchers seek to maximize individual reward, and they suggest implementing a ``Matthew Effect'', where researchers who are already credit-rich should be allocated more credit than in an even-split system.
Interestingly, this is coherent with the results of our paper, where it is socially optimal to obfuscate data so that researchers who are already ``ahead'' (in terms of data), end up ``ahead'' in terms of credit.
\iffalse
 
If multiple researchers solve the problem, the reward gets split between the researchers who solve the problem. 
Surprisingly, their results show that splitting this reward evenly (which seems like the fairest solution) is not always socially optimal when individuals want to maximize their expected rewards. 
Instead, the more ``unfair'' methods of assigning credit asymmetrically (giving some researchers more credit than what they would get under an even-split system, and others less credit), and giving credit that is out of proportion with the true importance of research projects, can lead to socially optimal results in equilibrium. It is interesting to note that in their paper, they suggest implementing a ``Matthew Effect'', where researchers who are already credit-rich should be allocated more credit than in an even-split system. This is coherent with the results of our paper, where it is socially optimal to obfuscate data so that researchers who are already ``ahead'' (in terms of data), end up ``ahead'' in terms of credit.
\fi

Cai, Daskalakis and Papadimitriou \cite{CaiDaskalakisPapadimitriou} study the problem of incentivizing $n$ players to share data, in order to compute a statistical estimator. Their goal is to minimize the sum of rewards made to the players, as well as the statistical error of their estimator. In contrast, our goal is to give a decentralized mechanism through which players can pool their data, and distribute partial information to themselves in order so as to increase the utility of every collaborating player. 

%The delayed release of data in MPC protocols is closely linked to the problem of ``timed-release crypto'' in general, which was introduced by \cite{May} and studied by \cite{RSW96} with their proposal of 
%\emph{time-lock puzzles} which would only release locked data after a certain time delay.
%Since the original proposal, there have been a number of works exploring this direction. We mention 
%timed commitments \cite{BN00} and the work of \cite{MMV13}. in the random oracle model. 

Boneh and Naor \cite{BN00} construct timed commitments that can be ``forced open'' after a certain time delay, 
and discuss applications of their timed commitments to achieve fair two-party contract signing (and coin-flipping) 
under certain timing assumptions including bounded network delay and the \cite{RSW96} assumption about sequentiality of modular exponentiation.

\paragraph{Roadmap}
Section \ref{sec:scientificCollaborationModel} covers the scientific collaboration model, mechanisms, and feasibility theorems.
Section \ref{sec:queuedDef} covers the definitions and constructions of ordered MPC, and
Section \ref{sec:timedDelayMPC} covers definitions and constructions of timed-delay MPC, and associated primitives such as time-line puzzles.

%\fi

%\input{intro-part2}

\section{Data sharing model}\label{sec:scientificCollaborationModel}

In this section, we present and analyze mechanisms for scientific collaboration in our model.
In our exposition, we focus primarily on the setting of scientific collaboration and publication. 
However, we want to highlight that our results apply to more broad collaboration and discovery in general, 
in which case a ``publication'' should be thought of as any kind of public output. 

\paragraph{Notation} We denote by $[n]$ the set $\{1,...,n\}$ of integers between $1$ and $n$, and by $[n] \to [n]$ the set of all permutations of $[n]$. 
For a set $X$, we write $\Delta(X)$ to denote the set of all distributions over $X$. The symbol $\sqcup$ denotes the disjoint union operation.
An \emph{efficient} algorithm is one which runs in probabilistic polynomial time (\PPT{}).

\subsection{The model}\label{sec:sharingModel} 
%We present a model of scientific collaboration, addressing \emph{whether} it is beneficial for parties to collaborate,
%and if so, \emph{how}: in particular, since the order of research publication matters, 
%we examine the problem of finding a feasible collaboration order.
%Once an order is found, the ordered and timed-delay MPC protocols of the previous sections may be used to securely implement beneficial collaborations.

We propose a model of collaboration between $n$ research groups which captures the following features.
Groups may pool their data, but each group will publish their own results. Moreover,
only results that improve on the ``state of the art'' may be published. That is, a new result must improve on prior publications.
However, more credit may be given to earlier publications.
Finally, a group will learn not only from pooling their data with other groups, but also from other groups' publications. 
\iffalse
\begin{itemize}
\item Groups may pool their data, but each group will publish their own results.
\item Only results that improve on the ``state of the art'' may be published. That is, a new result must improve on prior publications.
\item More credit is given to earlier publications. %Thus, the groups may have to agree in which order they want to publish.
\item A group will learn not only from pooling their data with other groups, but also from other groups' publications. Thus, any data-sharing protocol must take into account that information leaks as groups publish their results.
\end{itemize}
\fi

%\subsection{The model}

To formalize the intuitions outlined above, we specify a model as follows.

{\small
\begin{itemize}
\item There is a set $[n]$ of players.
\item Each player $i$ has a dataset $x_i$ which is sampled as follows. 
	\begin{itemize}
	\item For each $i\in[n]$, there is a set $X_i$ of possible datasets, which is common knowledge. Let $X$ denote $X_1\times\dots\times X_n$.
	\item There is a distribution $\cX\in\Delta(X)$ over $X$, from which the $x_i$ are sampled: $(x_1,\dots,x_n)\from\cX$. 
	\item The distribution $\cX$ is not known to any of the players, but comes from a commonly known distribution $\cD$. That is, $\cX\from\cD$, for some $\cD\in\Delta(\Delta(X))$.
	\end{itemize}
\item There is an output space $Y$, and a function $f:\Delta(X_1\times\dots\times X_n)\to Y$ such that $\hat{y}=f(\cX)$ is the value which the players wish to learn.
	That is, the players want to learn some property of the unknown distribution $\cX$ from which their datasets were sampled.
	$Y$ and $f$ are common knowledge.
%\item Let $\sparams$ denote the common-knowledge model parameters defined so far: 
	%$$\sparams=((X_i)_{i\in[n]},\cD,Y,f).$$
\item $\cY_0$ denotes the distribution of $\hat{y}$ given $f$ and $\cD$	.
\item There is a \emph{score function} $\score:\Delta(Y)\to\RR_+$, which varies with $f$ and $\cD$. 
	The score function $\score(\cdot)$ is maximized by the distribution $\hat{\cY}$ which puts probability $1$ on the true value $\hat{y}$.
	The score function $\score$ is common knowledge. 
	\begin{itemize}
	\item We require a natural \emph{monotonicity} property of the score function. Namely, let $\cY$ and $\cZ$ be any distributions, and let $z$ be a value in the support of $\cZ$. Then
	$$\score(\cY)\leq\score(\cY|z\from\cZ),$$
	where $z\from\cZ$ denotes the event that $z$ is sampled from the distribution $\cZ$.
	\item {\it Remark.} Let $\{\hat{y}|x_1,\dots,x_n\}$ denote the distribution of $\hat{y}$ given certain datasets $(x_1,\dots,x_n)\in X$.
	A consequence of the monotonicity condition is that given all of the datasets $x_1,\dots,x_n$ of all players in the model, 
	the best achievable score is $s\left(\{\hat{y}|x_1,\dots,x_n\}\right)$.
	\end{itemize}
\item A {\em collaboration outcome} is given by a permutation $\pi: [n] \to [n]$ 
	and a vector of output distributions $(\cZ_1,\dots,\cZ_n) \in (\Delta(Y))^n$ such that $\score(\cY_0)<\score(\cZ_{\pi(1)})<\dots<\score(\cZ_{\pi(n)})$.

	The intuition behind this condition is that, at time $t$, player $\pi(t)$ will publish $\cZ_{\pi(t)}$. 
	Since only results that improve on the ``state of the art'' can be published, we must have that the score $\score(\cZ_{\pi(t)})$ increases with the time of publication $t$.
\item  For a collaboration outcome $\omega = (\pi,\vec{\cZ})$, the player who publishes at time $t$ obtains a reward 
	$$R_t(\pi,\vec{\cZ}) = \beta^{t} \cdot (\score(\cZ_{\pi(t)})-\score(\cZ_{\pi(t-1)}))$$
	where $\beta \in (0,1]$ is a \emph{discount factor} which penalizes later publications.\footnote{This is motivated by market scoring rules \cite{Hansen},
	where experts are rewarded according to how much they improve existing predictions.} %The reward function is common knowledge.
\item For each player $i$, we define $\alpha_i=\score(\{\hat{y}|x_i\})-\score(\cY_0)\in\RR_+$, where $\{\hat{y}|x_i\}$ is the distribution of $\hat{y}$ given that the $i^{th}$ dataset is $x_i$.
	This models the ``outside payoff'' that player $i$ could get if she does not collaborate and simply publishes on her own.
\item Players may learn information not only from their own data, but also from the prior publications of others.
	A \emph{learning bound vector} $\{\lambda_{\pi,i}\}_{\pi\in([n]\rarr[n]),i\in[n]}$ characterizes, for any publication order $\pi$,
	the maximum amount that each player $i$ can learn from prior publications. This notion is defined formally in Section \ref{sec:dataSharingMechs}.
\item We define $\pparams$ to be the collection of all common-knowledge parameters of the model:
	$$\pparams=(\cD,f,\score,\beta).$$
\end{itemize}
}

\subsection{Examples}\label{sec:examples}

To illustrate the range of settings to which our model applies, 
we describe several concrete model instantiations.

Recall that our goal is to build mechanisms to enable collaborations by sharing data, 
in settings where such collaboration would be beneficial to all parties.
Intuitively,
such settings occur when the result that can be obtained based on the union of all players' datasets is ``much better''
than the results that can be obtained based on the individual datasets: in other words, 
the ``size of the pie'' to be split between the collaborating players 
is at least as large as the sum of the ``slices'' obtained by players working individually.
This intuition is made rigorous in Lemma \ref{clm:superadditive},
where we discuss score functions which satisfy a superadditivity condition (Property \ref{property:superadditive}).

\paragraph{Toy Example I: Secret-sharing.}
We begin with a ``toy example'' based on secret-sharing. This artificial first example 
is a dramatic illustration that the size of reward from collaboration can be much larger than 
the sum of individual rewards without collaborating.

Consider a stylized secret-sharing model with a secret $\hat{y}$ drawn uniformly at random from $\zo^n$.  Each player's data consists of a share $x_i \in \zo^{n}$ such that $\hat{y} = x_1 \xor ... \xor x_n$ be the secret the players are trying to reconstruct.  
The shares are correlated and drawn from a distribution $\cX$ as follows:
\begin{itemize}
\item For each $i \in [n-1]$, $x_i$ is uniformly random in $\zo^{n}$.
\item The last share is chosen such that $x_{n} = \hat{y} \xor x_1 \xor \dots \xor x_{n-1}.$
\end{itemize} 

The players want to learn $f(\cX) = \hat{y}$. The score from publishing  a distribution $\cY$ is 
$\score(\cY) = H(\hat{y})-H(\hat{y}| \cY)$
where $H(\hat{y}) = n$ is the entropy of the uniformly random string $\hat{y}$ and $H(\hat{y} | \cY)$ is the entropy of $\hat{y}$ given the distribution $\cY$.

Without collaborating, each player $i$ only knows a uniformly random string $x_i$. Thus, $H(\hat{y}|x_i) = H(\hat{y}) = n$ and $\alpha_i = H(\hat{y} | x_i)  - H(\hat{y}) = 0$ for each player $i$. Consider the following collaboration mechanism:
\begin{itemize}
\item Each player contributes share $x_i$ to the mechanism.
\item The mechanism computes $\hat{y} = x_1 \xor \dots \xor x_n$.
\item The mechanism reveals $i^{th}$ digit $\hat{y}_i$ to each player $i$. 
\end{itemize}

When participating in this mechanism, the first player will publish a guess $\cY_1$ which is a distribution over $\zo^n$ where the first bit of $y \from\cY_1$ is always $\hat{y}_1$.
All other players learn $\hat{y}_1$ from player 1's publication. 
Proceeding inductively, the $i^{th}$ player will publish a guess $\cY_i$ such that the first $i$ bits are correct, that is, $(y_1,\dots,y_i)=(\hat{y}_1,\dots,\hat{y}_i)$ for any $y \from\cY_i$.
Note that since $\alpha_i = 0$ for each player $i$, and $H(\hat{y} | \cY_i) - H(\hat{y} | \cY_{i-1}) = 1 > \alpha_i$, this mechanism incentivizes players to collaborate.

\iffalse
When participating in this mechanism, the first player will publish a guess $z_1 \in \zo^{n}$ such that $z_{11} = y_1$ and $z_{1j}$ is a random guess for all $j > 1$. All other players learn $y_1$ from player 1's publication. Proceeding inductively, the $i^{th}$ player will publish a guess $z_i \in \zo^{n}$ such that $z_{11},\dots,z_{1i} = y_1,...,y_i$ and $z_{ij}$ is random for all $j > i$. 
Since $\alpha_i = 0$ $\forall$ player $i$, and $s(z_i) - s(z_{i-1}) = 1 > \alpha_i$, players are incentivized to collaborate.
\fi

\paragraph{Toy Example II: Network flow.}
Let $G=(V,E)$ be a graph. % where $V$ is the set of vertices and $E$ is the set of edges. 
Let $\tilde{s},\tilde{t}\in V$ be vertices which are connected by some number of disjoint paths.
Consider a model where $V$, $\tilde{s}$, and $\tilde{t}$ are common knowledge, and each player's data consists of a disjoint subset of edges in $x_i\subseteq E$.
More precisely, $(x_1,\dots,x_n)\from\cX(E)$ where $\cX$ samples a partition of $E$.

The players want to learn the set of paths from $\tilde{s}$ to $\tilde{t}$. That is, $f(\cX(E))$ is the set of paths in $E$ from $\tilde{s}$ to $\tilde{t}$.
The score from publishing a distribution $\cZ$ over edges is
$$\score(\cZ)=|\{p:\mbox{$p$ is a path in $E$ from $\tilde{s}$ to $\tilde{t}$, and }\Pr_{z\from\cZ}\left[p\subseteq z\right]=1\}|.$$
In other words, the player's score is given by how many paths from $\tilde{s}$ to $\tilde{t}$ she knows with certainty to exist in $E$.
In some cases, it may be that no player knows any path from $\tilde{s}$ to $\tilde{t}$ based only on her own data, as illustrated by the simple example in the diagram below.
\begin{center}
\begin{tikzpicture}[every path/.style={>=latex},every node/.style={draw,circle},scale=0.8]
  \node            (s) at (0,0)  { S };
  \node            (a) at (1.5,0)  { ~ };
  \node            (t) at (1.5,-1.5) { T };
  \node        (b) at (0,-1.5) { ~ };
  \path[every node/.style={font=\rmfamily\small}]
    (s) edge[->] node [above] {$x_1$} (a)
    (s) edge[->] node [left] {$x_3$} (b)
    (a) edge[->] node [right] {$x_2$} (t)
    (b) edge[->] node [below] {$x_4$} (t);
\end{tikzpicture}
%\caption{A four-player collaboration example}
%\label{fig:flow}
\end{center}

Consider the following collaboration mechanism:
\begin{itemize}
\item Each player contributes their edges $x_i$ to the mechanism.
\item The mechanism computes $E=x_1\cup\dots\cup x_n$, and the set $P=\{p_1,\dots,p_k\}$ of paths in $E$ that start at $\tilde{s}$ and end at $\tilde{t}$.
\item The mechanism reveals the $i^{th}$ path $p_i$ to player $i$. If $k<n$, then the last $k-n$ players will get no output. 
	If $k>n$, the ``extra'' paths are allocated arbitrarily to players.\footnote{\label{ft:arbitrary1}It may be beneficial to allocate the ``extra'' paths strategically in order to reward players more fairly, 
	or in order to make collaboration possible when the outside option values $\alpha_i$ are nonzero. However, in this example, we allocate them arbitrarily for simplicity.}
\end{itemize}

When participating in this mechanism, the first player will publish a guess $\cZ_1$ which (always) samples the set $\{p_1\}$. All other players learn $p_1$ from player 1's publication.
Then, the $i^{th}$ player will publish a guess $\cZ_i$ that samples the set $\{p_1,\dots,p_i\}$.
As long as $\score(\cZ_i) - \score(\cZ_{i-1}) \geq \alpha_i$ for all $i\in[n]$ (note that this is the case in the diagram), this mechanism incentivizes players to collaborate.

\paragraph{Example III: Correlating gene loci with disease}

This example is inspired by successful GWAS studies to identify gene loci associated with schizophrenia.
Consider a model where each player holds a set of patients' medical (and in particular, genetic) data $x_i$ which comes from some unknown patient distribution $\cX$. 
The players wish to learn the set $f(\cX)$ of gene loci that are correlated with the occurrence of schizophrenia in patients.
%The players wish to learn which gene loci are correlated with the occurrence of schizophrenia in patients.
%More precisely, the target function $f(\cX)$ is a vector $((\gamma_1,p_1),\dots,(\gamma_N,p_N))$ where $\Gamma=\{\gamma_1,\dots,\gamma_N\}$ is the set of gene loci, and for each $j\in[N]$,
%$p_j$ is the correlation coefficient between $\gamma_1$ and occurrence of schizophrenia.

Let $\Gamma$ be the set of all gene loci. 
For $\gamma\in\Gamma$, define $\bI_\gamma$ to be 1 if $\gamma\in f(\cX)$ and 0 otherwise.
%For a distribution $\cZ$ over $\Gamma$, let $p_\cZ(\gamma)=\Pr_{z\from\cZ}[\gamma\in z]$ and let $q_\cZ(\gamma)=1-p_\cZ(\gamma)$.
The score from publishing a distribution $\cZ$ over $\powset(\Gamma)$ (i.e. over subsets of gene loci) could be:\footnote{In practice, 
a more realistic scenario might be to model the \emph{extent} to which
particular gene loci are found to be correlated with the occurrence of schizophrenia, 
rather than classifying into binary categories ``correlated'' and ``not correlated''. 
This case could be modeled, for example, 
by letting $f(\cX)$ be a vector $((\gamma_1,p_1),\dots,(\gamma_N,p_N))$ where $\Gamma=\{\gamma_1,\dots,\gamma_N\}$ 
is the set of gene loci, and for each $j\in[N]$,
$p_j$ is the correlation coefficient between $\gamma_1$ and occurrence of schizophrenia.
While Example IV presents the simpler ``binary'' model for ease of exposition, 
we remark that with appropriate modifications to the score function and mechanism, 
our model can accommodate the more complex case of estimating correlations, too.}
$$\score(\cZ)=\sum_{\gamma\in f(\cX)}\Pr_{z\from\cZ}[\gamma\in z] - \sum_{\gamma\notin f(\cX)}\Pr_{z\from\cZ}[\gamma\in z].$$
This score function rewards players for assigning high probabilities to gene loci $\gamma$ which are actually correlated with schizophrenia, 
and penalizes them for assigning high probabilities to those which are not.
As in our previous examples, it turns out that in this setting, the reward that can be obtained based on pooling all the players' data
is much greater than the sum of the rewards that could be obtained individually, as illustrated in Figure \ref{fig:gwas}.

\begin{figure}[ht!]
\centering
\ifdoublecolumn
\includegraphics[width=\columnwidth]{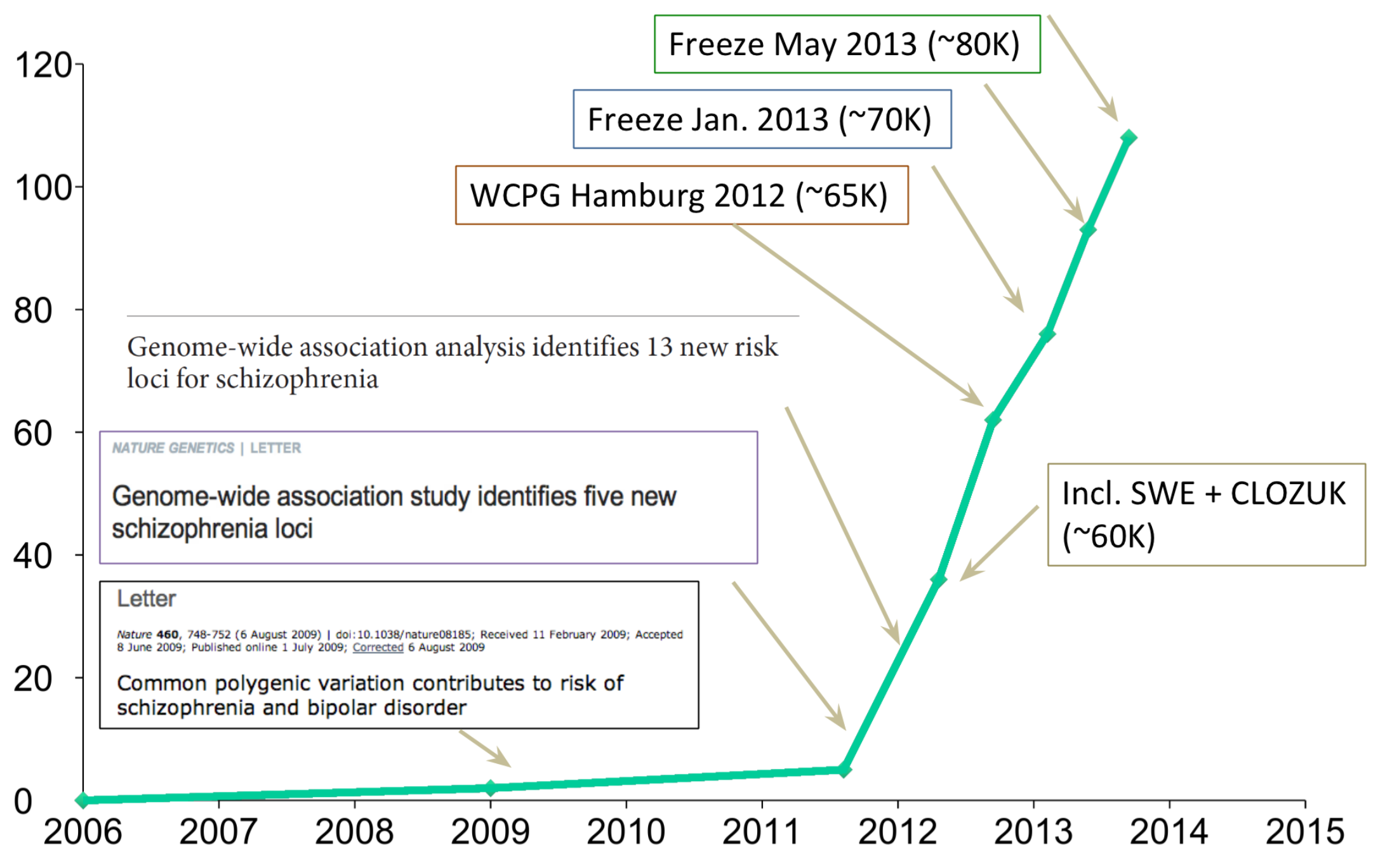}
\else
\includegraphics[width=0.7\textwidth]{schizo}
\fi
\caption{GWAS study success: the y-axis is the number of gene loci correlated with schizophrenia, 
and the x-axis is time (which corresponds to \emph{amount of data}, since the reason for the improved findings was accumulation of data over time).
Image \copyright Stephan Ripke}
\label{fig:gwas}
\end{figure}

Consider the following collaboration mechanism:\footnote{This is just one example of a reasonable mechanism for this model; we do not mean to claim that it is a canonical or optimal one. There are many variants which could make sense: for example, a simple modification would be to change the threshold $0.5$ in the second step.}
\begin{itemize}
\item Each player contributes some patient data $x_i$.
\item The mechanism computes $\cY^*=\{f(\cX)|x_1,\dots,x_n\}$, i.e. the distribution of $f(\cX)$ given all players' input data. 
	Let $\Gamma^*=\{\gamma\in\Gamma:\Pr_{y\from\cY^*}[\gamma\in y]>0.5\}$, that is, the set of gene loci that are more likely than not to be in $f(\cX)$, according to $\cY^*$.
\item The mechanism reveals to player $i$ the $i^{th}$ gene locus $\gamma_i$ in $\Gamma^*$. If $|\Gamma^*|<n$, then the last $k-n$ players will get no output. 
	If $|\Gamma^*|>n$, the ``extra'' gene loci are allocated arbitrarily.\footnote{As remarked in Footnote \ref{ft:arbitrary1}, 
	it can be beneficial to allocate the ``extra'' gene loci in a way which is not arbitrary, but instead optimized for making collaboration possible.
	In this example, for simplicity, we allocate them arbitrarily.}
\end{itemize}

When participating in this mechanism, the first player will publish a guess $\cZ_1$ which (always) samples the set $\{\gamma_1\}$. All other players learn $\gamma_1$ from player 1's publication.
Then, the $i^{th}$ player will publish a guess $\cZ_i$ that samples the set $\{\gamma_1,\dots,\gamma_i\}$.
Provided that $\score(\cZ_i) - \score(\cZ_{i-1}) \geq \alpha_i$ for all $i\in[n]$ (note that Figure \ref{fig:gwas} depicts exactly such a scenario), this mechanism incentivizes players to collaborate.

\paragraph{Example IV: Statistical estimation}
Our last example is one where -- in contrast to the examples so far -- there are \emph{decreasing} marginal returns from adding new information, and thus collaboration will not be feasible.

We consider a simple Bayesian model where the distribution $\cX$ is itself drawn from a ``distribution over distributions'' $\cD$. More concretely, each player $i$ receives a vector of $k_i$ samples $(x_{i,1},...,x_{i,k_i})$ drawn independently from a normal distribution $N(\mu,\sigma^2)$ with unknown mean $\mu$ and known variance $\sigma^2$. The mean $\mu$ is itself drawn from a commonly known prior distribution $\cD = N(m,1)$ with known mean $m$ and variance 1.
In this case, the ground set $X_i$ is $\bR^{k_i}$. The distribution $\cX(\mu,\sigma)$ is a product distribution over $\bR^{\sum_{i=1}^n k_i}$, where each component of $(x_{1,1},...,x_{n,k_n})$ is drawn independently from $N(\mu,\sigma)$. The players want to learn $f(\cX(\mu,\sigma))= \mu$. 

An estimator for $\mu$ is a random variable $\hat{\mu}$. The score of such a guess $\hat{\mu}$ is $\score(\hat{\mu})=  -\bE[(\hat{\mu} - \mu)^2]$. It is well known that if we have a vector $(x_{i,1},...,x_{i,k_i})$ of  random samples drawn from $N(\mu,\sigma)$, the estimator that minimizes the expected squared error to $\mu$ is  $\hat{\mu_i} = \frac{1}{k_i} \sum_{j=1}^{k_i} x_{i,j}$. Note that this is a normal random variable since each $x_{i,j}$ is sampled from normal random variable. The expectation of $\hat{\mu_i}$ is $\frac{1}{k_i}\cdot k_i \cdot\mu = \mu$ and the variance of $\hat{\mu_i}$ is $\frac{1}{k_i^2}\cdot k_i \cdot\sigma^2 = \frac{1}{k_i}\cdot \sigma^2$. Thus, $\score(\hat{\mu_i}) = \frac{1}{k_i}\cdot\sigma^2$. If a player published by herself and did not collaborate, her reward would be the difference 
$\alpha_i = \sigma^2 - \frac{1}{k_i} \cdot\sigma^2$
between the priorly known variance $\sigma^2$ and the variance $\frac{1}{k_i}\cdot \sigma^2$ of player $i$'s estimate. 

If the players collaborate, they can obtain the estimator $\hat{\mu^*}  =\frac{1}{\sum_{i=1}^n k_i} \sum_{i=1}^n \sum_{j=1}^{k_i} x_{i,j} $ which has variance  $\score(\hat{\mu^*}) = \frac{1}{\sum_{i=1}^n k_i}\sigma^2$. The reward for $\hat{\mu}^*$ is the reduction in variance  $\sigma^2 - s(\hat{\mu^*}) = \sigma^2 \cdot (1- \frac{1}{\sum_{i=1}^n k_i})$.
Note that in this case, the reward from an estimator only depends on the number of data points $N$  used to construct this estimator (in the above notation, $N = \sum_{i=1}^n k_i$). Furthermore, the reward $R(N) = \sigma^2 (1- \frac{1}{N})$ that one could obtain with $N$ data points is concave in $N$. Intuitively, if one only has $N=2$ data points, and gets 10 new ones, those 10 new data points are very valuable. However, if one already has $N = 2000000$ data points and gets 10 new ones, those 10 new data points do not increase the score very much. 

This setting is in contrast to our Example III, where the score seemed to increase in a convex way with the number of data points. Indeed, in this Bayesian example, we will always have that  
$$R(\sum_{i=1}^n k_i)= \sigma^2 (1- \frac{1}{\sum_{i=1}^n k_i}) \leq   \sigma^2 \sum_{i=1}^n (1-\frac{1}{k_i}) = \sum_{i=1}^n R(k_i).$$

In Section 2.5 we elaborate on why the above inequality is bad for collaboration. Intuitively,
the left-hand side is the ``size of the pie'' if all players were to collaborate, 
and the right-hand side is the sum of the rewards that each player could receive on her own.  
The inequality implies there is no way to ``slice the pie'' so that every player has a bigger reward than the 
$\alpha_i$ they can get without collaborating, and thus collaboration is impossible. 

In this simple Bayesian example, the marginal value of extra information will be decreasing. 
This raises the interesting question of \emph{when} 
the value of information is (and is not) not convex with the amount of information available.
For example, consider machine learning: learning problems whose objectives can be stated as minimizing a convex loss 
function (or maximizing a concave value function) seem to induce natural score functions 
which do not have increasing marginal returns, so
our model may be more applicable to problems with non-convex objectives. 
We remark that such non-convex learning problems, in which our model seems more applicable, 
are an area of interest in machine learning as solving them is lately becoming practical -- we refer to
Bengio and LeCun \cite{LeCun} for a more thorough discussion of this situation.

\subsection{Data-sharing mechanisms}\label{sec:dataSharingMechs}

We now return to the general formulation of our collaboration model, and
we seek to design a general data-sharing mechanism that takes as input the data of all the parties, 
computes an output distribution $\cY_i\in\Delta(Y)$ for each $i\in[n]$, and outputs $\cY_i$ to each player $i$.
The mechanism will output the $\cY_i$ values to players sequentially, in a particular order.
Upon receiving $\cY_i$, player $i$ produces a public output (i.e a publication in the research collaboration example) which we denote by $\cZ_i\in Y$.  

We note that the public output of player $i$ will not necessarily be the same as what was delivered by the data-sharing mechanism. 
Since player $i$ wants to maximize her reward, she will publish a result $\cZ_i$ that will maximize her reward, conditional on the information she has at the time of publication. 
This information includes, in addition to the output $\cY_i$ which she receives from the mechanism (and her knowledge of how the mechanism works\footnote{The mechanism description is common knowledge.}), 
also her own dataset $x_i\in X_i$, and all the outputs $\cZ_j$ of other players that published before her. 

Recall that a {\em collaboration outcome} $(\pi,\vec{\cZ})$ is given by a permutation $\pi: [n] \to [n]$ and a vector of output distributions $\vec{\cZ}=(\cZ_1,\dots,\cZ_n) \in (\Delta(Y))^n$ 
such that $\score(\cY_0)<\score(\cZ_{\pi(1)})<\dots<\score(\cZ_{\pi(n)})$. We now define a \emph{proposed collaboration outcome} $(\pi,\vec{\cY})$ 
as a permutation $\pi:[n]\to[n]$ together with a vector of \emph{proposed} outputs $\vec{\cY}=(\cY_1,\dots,\cY_n)\in (\Delta(Y))^n$ generated by a data-sharing mechanism,
satisfying $\score(\cY_0)<\score(\cY_{\pi(1)})<\dots<\score(\cY_{\pi(n)})$.

Recall also that we need to bound  how much player $i$ can learn from previous publications (and from her own dataset). 
We formally capture this with the notion of {\it learning bound vectors} $\lambda_{\pi,i}$, which give an upper bound on the amount that player $i$ learns 
from all previous publications when the order of publication is determined by permutation $\pi$.

\begin{definition}\label{def:learningBoundVector}
A {\em learning bound vector} 
\ifdoublecolumn
$$\vec{\lambda}=(\lambda_{\pi,i})_{\pi \in ([n]\rightarrow[n]), i \in [n]}$$
\else
$\vec{\lambda}=(\lambda_{\pi,i})_{\pi \in ([n]\rightarrow[n]), i \in [n]}$ 
\fi
is a non-negative vector such that, if $(\pi,\vec{\cY})$ is a collaboration outcome proposed by a data-sharing mechanism, and $\cZ_i$ is the best (i.e. highest-scoring) distribution that player $i$ can compute at the time $\pi^{-1}(i)$ of her publication,
then $\score(\cZ_i) \leq \score(\cY_i) + \lambda_{\pi,i}$.
%Note that player $i$ will compute $\cZ_i$ based on her knowledge of how the mechanism works, and on the sample $y_i\from\cY_i$ that she received as output from the mechanism,
%and on the prior publications $z_{\pi(1)},\dots,z_{\pi^{-1}(i)-1}$ of other players.
%
Let $\Lambda = \bR^{n! \times n}_{+}$ denote the set of all learning bound vectors. 
\end{definition}

\begin{definition}\label{def:inferredOutputs}
For a learning bound vector $\vec{\lambda}$,
the set of {\em inferred output distributions} derived from a proposed collaboration outcome $(\pi,\vec{\cY})$ is given by the following expression:
\ifdoublecolumn
\begin{gather*}
\cI_{\vec{\lambda}}(\pi,\vec{\cY}) = \\ \{(\cZ_1,\dots,\cZ_n): \forall t\in[n],~ \score(\cY_{\pi(t)}) \leq \score(\cZ_{\pi(t)}) \leq \score(\cY_{\pi(t)}) + \lambda_{\pi,\pi(t)}\}.
\end{gather*}
\else
$$\cI_{\vec{\lambda}}(\pi,\vec{\cY}) = \{(\cZ_1,\dots,\cZ_n): \forall t\in[n],~ \score(\cY_{\pi(t)}) \leq \score(\cZ_{\pi(t)}) \leq \score(\cY_{\pi(t)}) + \lambda_{\pi,\pi(t)}\}.$$
\fi
\end{definition}

The intuition behind the above definition is that the amount of information that player $\pi(t)$ (namely, the player who publishes at time $t$) can learn from prior outputs 
is measured by how much her score increases based on these prior outputs.
%By the definition of our vector learning bounds, this 
This increase in score is bounded by $\lambda_{\pi,\pi(t)}$. 
Thus, her eventual output will be some $\cZ_{\pi(t)}$ with score between $\score(\cY_{\pi(t)})$ and $\score(\cY_{\pi(t)})+\lambda_{\pi,\pi(t)}$. 

%For example, when outputs are random variables and the distance is the expected squared error to the true answer, we have that the best possible inferred output
%for player $\pi(t)$ is $z_{\pi(t)}=\min_z\E[(z-y^*)^2|x_{\pi(t)},y_{\pi(t)},z_{\pi(t-1)},\dots,z_{\pi(1)},y_0]$. In this case, $\lambda_{\pi,\pi(t)}=d(y_{\pi(t)},y^*)-d(z_{\pi(t)},y^*)$.

\begin{remark}
In certain cases, $\lambda_{\pi,\pi(t)}$ measures exactly the amount of information that player $\pi(t)$ can learn from her data. However, in our definition $\lambda_{\pi,\pi(t)}$ is an upper bound, and we emphasize that it may be a loose upper bound on the amount of information $\pi(t)$ can learn. Our emphasis on this point comes from the following two reasons.
\begin{itemize} 
\item In general, the vector $\vec{\lambda} \in \bR^{n! \times n}$ has very high dimension, and finding such a vector is infeasible. We may want to approximate this vector via a low-dimensional encoding (as we will do below, where we encode learning bounds using $n$-dimensional vectors). Since this low-dimensional encoding will lose information, we will not be able to represent $\lambda_{\pi,\pi(t)}$ exactly, but may get a reasonable upper bound on its value.
\item For some other settings, we may not be able to derive a precise expression for $\lambda_{\pi,\pi(t)}$ in terms of expectations, but we may still be able to derive an upper bound on the amount of information that player $\pi(t)$ learns. 
\end{itemize}
\end{remark}

Now that we have established a formal definition of learning bound vectors, we proceed to formally define a data-sharing mechanism.

\begin{definition}\label{def:dataSharingMech}
For model parameters $\pparams$, a {\em data sharing mechanism} is a function 
$$M: X \times \Lambda \to ([n]\rightarrow[n]) \times (\Delta(Y))^n$$
which takes as inputs a vector $\vec{x}=(x_1,\dots,x_n)$ of datasets 
%a vector $(\alpha_1,\dots,\alpha_n)$ of outside options, 
and
$\vec{\lambda}=(\lambda_{\pi,i})_{\pi \in ([n]\rightarrow[n]), i \in [n]}$ a learning bound vector,
and outputs an ordering $\pi$ of the players and an output vector $(\cY_1,\dots,\cY_n)\in (\Delta(Y))^n$.
\end{definition}

\begin{remark}
In the definition, for the sake of generality, we assume that the $\vec{\lambda}$ values are given as input to the mechanism.
We remark that in certain settings, these values can be computed directly from the inputs $x_i$ of the parties, 
as discussed in the examples of Section 1.1.2.
In this case, one may think of the mechanism $M:X\to ([n]\rightarrow[n]) \times (\Delta(Y))^n$ as having input domain $X$ only.
\end{remark}

\subsection{Collaborative equilibria}\label{sec:collabEquil}

In our model, each research group $\pi(t)$ will collaborate only if the credit they obtain 
from doing so is greater than the ``outside option'' reward $\alpha_{\pi(t)}$. 
We want to design a mechanism that guarantees collaboration whenever possible.
Accordingly, we define the following equilibrium concept. 

\begin{definition}\label{def:collabEquil}
Let $\pparams$ be the model parameters.
Let $(\vec{x},\vec{\lambda}) \in X \times \Lambda$ and let $(\pi,(\cY_1,\dots,\cY_n)) \in ([n]\rightarrow[n]) \times (\Delta(Y))^n$. 
%Let $\cI(\pi,\mathbf{y})$ be the set of inferred output distributions from $(\pi,\mathbf{y})$. 
We say that $(\pi,(\cY_1,\dots,\cY_n))$ is a \emph{collaborative equilibrium with respect to $(\vec{x},\vec{\lambda})$}
if for all inferred output distributions $\vec{\cZ}=(\cZ_1,\dots,\cZ_n) \in \cI(\pi,(\cY_1,\dots,\cY_n))$ and all $t\in[n]$,
it holds that $R_t(\pi,\vec{\cZ})\geq\alpha_{\pi(t)}$. %where $R_t$ is the reward function as defined in Section \ref{sec:sharingModel}.
%$\beta^t \cdot (d(z_{\pi(t-1)},y^*) - d(z_{\pi(t)},y^*)) \geq \alpha_{\pi(t)}$.
\end{definition}

%\begin{remark} 
%Note that in the proposed game, each player $\pi(t)$ is honest about their data $x_{\pi(t)}$. The player thus has only two actions, collaborate or not collaborate. Not collaborating gives a payoff of $\alpha_{\pi(t)}$. As long as everybody else is collaborating, collaboration gives player $\pi(t)$ a payoff in the set 
%$$\{\beta^t (Var(z_{\pi(t-1)})-Var(z_{\pi(t)})): \mathbf{z} \in \cI(\pi,\mathbf{y})\}.$$
%The assumption that all other players are collaborating implies that our solution concept is a Nash equilibrium (as opposed to, for example, dominant strategies).
%Since we cannot know which payoff player $\pi(t)$ will obtain in this set, our definition of equilibrium will seek
%to guarantee that the player always wants to collaborate, regardless of which $\mathbf{z}$ is chosen. 
%\end{remark}

%Our goal is to find data-sharing mechanisms $M_f$ for which collaboration is an equilibrium. 
%Since such mechanisms will output permutations $\pi$ that satisfy some equilibrium constraints, 
%it shouldn't be surprising that the general problem is $\NP$-complete. 
%However, there is a very large subset of learning vectors for which we can efficiently find 
%(if it exists) a mechanism $M_f$ for which collaboration is an equilibrium. 

Our goal is to find data-sharing mechanisms for which collaboration is an equilibrium.  
Intuitively, since we are searching for a feasible permutation over a very high-dimensional space ($n!$-dimensional, to be precise), the problem will be $\NP$-complete 
(this is proven in Theorem \ref{thm:NPC}). 
However, there is a very natural condition on the learning vectors for which we can reduce the dimension of the search space and efficiently find a collaborative equilibrium.
The feasible case corresponds to the case where, for any player $j$, there is a bound on the amount of information that player $j$ could {\em teach} any other players. We denote this bound by $\mu_j$. Analogously, we could define $\mu_j$ to be a bound on the amount that player $j$ can {\em learn} from any other player. In this work, we describe only the first case, when $\mu_j$ represents a bound on how much information player $j$ can teach other players. The other case is analogous. 

We define a learning bound vector to be $\rankOne$ if it satisfies the following property.

\begin{definition}
A learning vector $\vec{\lambda} \in \Lambda$ is \emph{$\rankOne$} if there is a non-negative vector $(\mu_1,\dots,\mu_n)$ such that
$\lambda_{\pi,\pi(t)} = \sum_{\tau=1}^{t-1} \mu_{\pi(\tau)}$.
Let $\Lambda_1 \subset \Lambda$ denote the set of all $\rankOne$ learning vectors. 
\end{definition}

When $\vec{\lambda}$ is an $\rankOne$ learning vector, the total amount that player $\pi(t)$ learns from all prior outputs is $\sum_{\tau=1}^{t-1} \mu_{\pi(\tau)}$. 
In this case, we can give necessary and sufficient conditions for an equilibrium to exist (detailed in Theorem \ref{thm:necsuf} below), 
provided that the following Output Divisibility Condition is satisfied.

%Note that since the output $y^*$ is real-valued, 
%it is possible to add noise to this value $y^*$ to obtain less accurate approximations $y_i=y^*+Normal(0,\delta_i)$.
%More precisely, the following ``divisibility'' condition is satisfied, which is important for the following analysis.

\paragraph{Output Divisibility Condition.}
Given the model parameters $\pparams$ and any real $0<\delta\leq 1$,\footnote{Recall (from the model description) that $\score(\{\hat{y}|\cX\})=\max_{\cY\in\Delta(Y)}(\score(\cY))$. 
Without loss of generality, we assume in our analysis that the score function is normalized so that its maximum value $\score(\{\hat{y}|\cX\})=1$.} 
there exists a distribution $\cY\in\Delta(Y)$ such that $\score(\cY)=\delta$.

\begin{remark}\label{rmk:outputDiv}
The Output Divisibility Condition holds for a wide variety of natural score functions.
In general, score functions which reward ``how close'' a distribution is to the true value $\hat{y}=f(\cX)$ will decrease (continuously) with the addition of random noise to a distribution.
Provided that this holds, the Output Divisibility Condition can be satisfied by taking the optimal distribution $\{\hat{y}|\cX\}$ and perturbing it with random noise:
the exact amount of noise to be added depends on the desired value of $\delta$.
To give a concrete example: in Example III (Gene loci), 
the perturbed distribution could simply add noise to the probabilities that each gene locus is sampled.
Here, ``adding noise'' can mean simply adding some $\eta\from N(0,\sigma^2)$ to the relevant parameters, where the magnitude of $\sigma$ depends on the 
precise formulation of the score function and the desired value of $\delta$.
\end{remark}

\begin{theorem}
\label{thm:necsuf}
Suppose that the Output Divisibility Condition holds.
Let $\vec{x}$ be a vector of inputs and $\vec{\lambda}$ be an $\rankOne$ learning bound vector. Let $\lambda_{\pi,\pi(t)} =\sum_{\tau=1}^{t-1} \mu_{\pi(\tau)}$. 
Then for $(\pi,\vec{\cY})$ to be a collaborative equilibrium, it is necessary and sufficient that
$$\sum_{t=1}^n \frac{\alpha_{\pi(t)}}{\beta^t} + \sum_{t=1}^n (n-t)\mu_{\pi(t)} \leq \score(\cY_{\pi(n)}) - \score(\cY_0).$$
\end{theorem}

\begin{proof}
\textbf{Necessity.} Let $(\pi,\vec{\cY})$ be a proposed collaborative equilibrium, and let $\vec{\cZ} \in \cI(\pi,\vec{\cY})$ be a possible vector of inferred outputs. 
For every $t$, we must have that:
$$\beta^t\cdot(\score(\cZ_{\pi(t)}) - \score(\cZ_{\pi(t-1)})) \geq \alpha_{\pi(t)}.$$
This is equivalent to:
$$\score(\cZ_{\pi(t)}) - \score(\cZ_{\pi(t-1)}) \geq \frac{\alpha_{\pi(t)}}{\beta^t}.$$
The worst case for player $\pi(t)$ is when player $\pi(t-1)$ learns as much as possible from prior publications and player $\pi(t)$ learns as little as possible. That is, when
\ifdoublecolumn
\begin{gather*}
\score(\cZ_{\pi(t-1)}) = \score(\cY_{\pi(t-1)}) + \mu_{\pi(1)} + \dots + \mu_{\pi(t-2)} \\
\mbox{ and } \score(\cZ_{\pi(t)}) = \score(\cY_{\pi(t)}).
\end{gather*}
\else
\begin{gather*}
\score(\cZ_{\pi(t-1)}) = \score(\cY_{\pi(t-1)}) + \mu_{\pi(1)} + \dots + \mu_{\pi(t-2)}\qquad\mbox{ and }\qquad\score(\cZ_{\pi(t)}) = \score(\cY_{\pi(t)}).
\end{gather*}
\fi
In this case, the equilibrium condition becomes: 
$$\score(\cY_{\pi(t)}) - \score(\cY_{\pi(t-1)}) - \sum_{\tau=1}^{t-2} \mu_{\pi(\tau)} \geq \frac{\alpha_{\pi(t)}}{\beta^t}.$$
Rearranging slightly, we obtain:
$\score(\cY_{\pi(t-1)}) - \score(\cY_{\pi(t)}) \leq  -  \frac{\alpha_{\pi(t)}}{\beta^t} -  \sum_{\tau=1}^{t-2} \mu_{\pi(\tau)}$.
Let us abuse notation slightly and define $\pi(0)=0$.
Then, summing over all $t$ yields
$$\score(\cY_{\pi(0)}) - \score(\cY_{\pi(n)}) \leq -\sum_{t=1}^n \frac{\alpha_{\pi(t)}}{\beta^t} - \sum_{t=1}^n (n-t)\mu_{\pi(t)}.$$
Flipping the signs in the inequality, the existence of a collaborative equilibrium implies:
$$\sum_{t=1}^n \frac{\alpha_{\pi(t)}}{\beta^t} + \sum_{t=1}^n (n-t)\mu_{\pi(t)} \leq \score(\cY_{\pi(n)}) - \score(\cY_0).$$

\textbf{Sufficiency. } To prove that the condition is sufficient: given  $\cY_{\pi(n)}$ satisfying the inequality in the theorem statement,
we need to construct $\vec{\cY}=(\cY_1,\dots,\cY_n)$ such that $(\pi,\vec{\cY})$ is a collaborative equilibrium. We construct $\vec{\cY}$ inductively as follows:
let $\delta_{\pi(n)} = \score(\cY_{\pi(n)})$, and
for any $t$ such that $2 \leq t \leq n$, let $\delta_{\pi(t-1)} = \delta_{\pi(t)} - \frac{\alpha_{\pi(t)}}{\beta^t} - \sum_{\tau=1}^{t-2} \mu_{\pi(\tau)}$.
Now that we have defined $\{\delta_{\pi(t)}\}_{t=1}^n$ in this way, it follows that if we set $\cY_{\pi(t)}$ such that $\score(\cY_{\pi(t)}) = \delta_{\pi(t)}$, then for all $t \geq 2$ we have
$$\score(\cY_{\pi(t)}) - \score(\cY_{\pi(t-1)}) = \delta_{\pi(t)} - \delta_{\pi(t-1)} = \frac{\alpha_{\pi(t)}}{\beta^t} + \sum_{\tau=1}^{t-2} \mu_{\pi(\tau)}.$$
Note that it is possible to set $\cY_{\pi(t)}$ in the required way, by the Output Divisibility Condition.
Rearranging the above equation, it follows that:
$$\beta^t \cdot(\score(\cY_{\pi(t)})- \score(\cY_{\pi(t-1)})-\sum_{\tau=1}^{t-2} \mu_{\pi(\tau)}  ) = \alpha_{\pi(t)}.$$
Since for any inferred outcome $\cZ_{\pi(t)}$ we have (by the definition of the learning bound vector) that
$$\score(\cY_{\pi(t)})\leq \score(\cZ_{\pi(t)}) \leq \score(\cY_{\pi(t)})+\lambda_{\pi,\pi(t-1)} =  \score(\cY_{\pi(t)})+\sum_{\tau=1}^{t-2} \mu_{\pi(\tau)},$$ 
we conclude that for all $t \geq 2$,
$$\beta^t \cdot(\score(\cZ_{\pi(t)}) - \score(\cZ_{\pi(t-1)}) ) \geq \alpha_{\pi(t)}.$$

Finally, we need to check that player $\pi(1)$ is incentivized to collaborate. 
Note that player $\pi(1)$ publishes first, so she cannot learn anything from previous publications. She will be incentivized to publish if 
$$\beta \cdot(\delta_{\pi(1)} - \score(\cY_0)) \geq \alpha_{\pi(1)}.$$
This condition is equivalent to
$$\delta_{\pi(1)} - \score(\cY_0) \geq \frac{\alpha_{\pi(1)}}{\beta}.$$
Replacing $\delta_{\pi(t-1)} = \delta_{\pi(t)} - \frac{\alpha_{\pi(t)}}{\beta^t} - \sum_{\tau=1}^{t-2} \mu_{\pi(\tau)} $ iteratively, we get that player $\pi(1)$ is incentivized to collaborate if and only if
$$\delta_{\pi(n)} - \score(\cY_0) \geq \sum_{t=1}^n \frac{\alpha_{\pi(t)}}{\beta^t} + \sum_{t=1}^n (n-t)\mu_{\pi(t)}$$
which is guaranteed by assumption.  
\end{proof}

Recall from the definition of the score function that the best score that can be attained given datasets $x_1,\dots,x_n$ is equal to $\score(\{\hat{y}|x_1,\dots,x_n\})$.
Based on Theorem \ref{thm:necsuf}, we can now characterize the datasets and learning bound vectors for which a collaborative equilibrium is possible.

\begin{definition}\label{def:supportCollabEquil}
Let $\pparams$ be the model parameters and let $(\vec{x},\vec{\lambda}) \in X \times \Lambda$.
We say that $(\vec{x},\vec{\lambda})$ \emph{supports a collaborative equilibrium} if it holds that
$$\sum_{t=1}^n \frac{\alpha_{\pi(t)}}{\beta^t} + \sum_{t=1}^n (n-t)\mu_{\pi(t)} \leq \score(\{\hat{y}|x_1,\dots,x_n\}) - \score(\cY_0).$$
\end{definition}

\subsubsection{How do the model parameters affect feasibility of collaborative equilibria?}

Consider for a moment the simple case where $\beta=1$ and $\vec{\lambda}=\vec{0}$, that is, there is no discount factor and players do not learn from others' publications.
%Now that we have established Definition \ref{def:supportCollabEquil}, 
We can show that in this case,
if the score function satisfies the following Property \ref{property:superadditive},
then it holds that for \emph{all} $\vec{x}\in X$, $(\vec{x},\vec{\lambda})$ supports a collaborative equilibrium. 
That is, in this simple case, the condition for $(\vec{x},\vec{\lambda})$ to support an equilibrium reduces to 
the superadditivity of the auxiliary score function $\overline{\score}$ given in Property \ref{property:superadditive}.

\begin{definition}\label{def:superadditive}
Let $S$ be a set.
A function $f:S\to\RR$ is \emph{superadditive} if for all disjoint $S_1,S_2\subseteq S$, it holds that
$f(S_1)+f(S_2)\leq f(S_1\cup S_2)$.
\end{definition}

\begin{property}[Superadditive Differences]\label{property:superadditive}
Let $\pparams$ be the model parameters. We define an \emph{auxiliary score function} $\overline{\score}:X_1\sqcup\dots\sqcup X_n\to\RR_+$ which maps a set of datasets to a real-valued score, as follows:
$$\overline{\score}\left(\{(i_1,x_{i_1}),\dots,(i_k,x_{i_k})\}\right)=\score(\{\hat{y}|x_{i_1},\dots,x_{i_k}\})-\score(\cY_0),$$
where $\{\hat{y}|x_{i_1},\dots,x_{i_k}\}$ denotes the distribution of $\hat{y}$ given that the datasets $x_{i_1},\dots,x_{i_k}$ were sampled\footnote{More precisely: 
$\{\hat{y}|x_{i_1},\dots,x_{i_k}\}$ is the distribution of $\hat{y}$ given that each $x_{i_j}$ was sampled in the ${i_j}^{th}$ position. 
(Recall that the distribution $\cX$ is over tuples of datasets $(x_1,\dots,x_n)$.)} from $\cX$.
The score function $\score$ satisfies the Superadditive Differences Property if $\overline{\score}$ is superadditive.
\end{property}

We observe that this precisely captures the intuition initially described in Section \ref{sec:examples}, 
that our model is designed to promote collaboration in situations where the reward that can be obtained
from pooling all players' data is more than the sum of the individual rewards that players can get.

\begin{lemma}\label{clm:superadditive}
Let $\pparams$ be model parameters such that $\beta=1$, let $\vec{x}\in X$ be arbitrary, and let $\vec{\lambda}=\vec{0}\in\Lambda$.
If $\overline{\score}$ is a superadditive function on the input data, then $(\vec{x},\vec{\lambda})$ supports a collaborative equilibrium.
\end{lemma}
\begin{proof}
Recall the inequality from Definition \ref{def:supportCollabEquil}:
$$\sum_{t=1}^n \frac{\alpha_{\pi(t)}}{\beta^t} + \sum_{t=1}^n (n-t)\mu_{\pi(t)} \leq \score(\{\hat{y}|x_1,\dots,x_n\}) - \score(\cY_0).$$

Since $\beta=1$ and $\mu_{\pi(t)}=0$, the left-hand side is simply $\sum_{t=1}^n \alpha_{\pi(t)}$. Using the definitions of $\alpha_{\pi(t)}$ and $\overline{s}$,
and the fact that $\pi$ is a permutation, this can be rewritten as:
$$\sum_{t=1}^n \alpha_{\pi(t)} = \sum_{t\in[n]}\left(\score(\{\hat{y}|x_{\pi(t)}\})-\score(\cY_0)\right) = \sum_{i\in[n]}\overline{\score}(\{(i,x_i)\}).$$

Substituting back into the inequality, we obtain:
$$\sum_{i\in[n]}\overline{\score}(\{(i,x_i)\}) \leq \score(\{\hat{y}|x_1,\dots,x_n\})) - \score(\cY_0).$$

The right-hand side of the inequality is, by definition, equal to $\overline{\score}(\{(1,x_1),\dots,(n,x_n)\})$.
Thus, the superadditivity of $\overline{\score}$ implies that the inequality holds, and it follows that $(\vec{x},\vec{\lambda})$ supports a collaborative equilibrium.
\end{proof}

Finally, we remark that either decreasing the discount factor $\beta$ or increasing the learning bound vector $\vec{\lambda}$
will make it harder to support a collaborative equilibrium 
(i.e. a lower value of $\beta$ means there will be fewer $(\vec{x},\vec{\lambda})$ which support an equilibrium),
since these cause the left-hand side of the inequality to increase.
So, while superadditivity is a sufficient condition in the simplest case, we observe that
determining which $(\vec{x},\vec{\lambda})$ support a collaborative equilibrium is a more complex problem when the model parameters are varied.

\subsection{The polynomial-time mechanism}
We show a polynomial-time mechanism that computes a collaborative equilibrium
in the case that learning bounds are given by a $\rankOne$ vector, provided that
the following \emph{Efficient} Output Divisibility Condition is satisfied. 
The Efficient Output Divisibility Condition is a natural extension of the Output Divisibility Condition,
which requires not only existence but also efficient computability of distributions with arbitrary score,
while taking into account that the best possible score for given input datasets $x_1,\dots,x_n$ is equal to $\score(\{\hat{y}|x_1,\dots,x_n\})$.

\paragraph{Efficient Output Divisibility Condition.}
Given model parameters $\pparams$, datasets $x_1,\dots,x_n\in X$, and any real $0<\delta<\score(\{\hat{y}|x_1,\dots,x_n\})$, 
it is possible to efficiently compute a distribution $\cY\in\Delta(Y)$ such that $\score(\cY)=\delta$.

\begin{remark}
The above condition holds for a wide variety of score functions, too:
in particular, it holds for the class of score functions described in Remark \ref{rmk:outputDiv}.
Suppose that the score function is continuous and decreases with the addition of random noise to a distribution.
Then the condition can be satisfied by taking the ``best computable'' distribution $\{\hat{y}|x_1,\dots,x_n\}$ and perturbing it with random noise:
the amount of noise to add will depend on the desired value of $\delta$.
\end{remark}

\begin{theorem}\label{thm:mechanismExists}
Suppose the Efficient Output Divisibility Condition holds.
Then there is a polynomial-time mechanism $\ShareData : X \times \Lambda_1$ that, 
%$M_f : X \times A \times \Lambda_1$ that, 
given inputs $(\vec{x},\vec{\mu})$ where $\vec{\mu}=(\mu_1,\dots,\mu_n)$ represents a $\rankOne$ learning vector, 
outputs a collaborative equilibrium $(\pi,\vec{\cY})$ whenever an equilibrium is supported by the inputs $(\vec{x},\vec{\mu})$ (as defined in Definition \ref{def:supportCollabEquil}),
and outputs ${\tt NONE}$ otherwise.
\end{theorem}

\begin{algorithm}[ht!]
  	\caption{$\ShareData((x_1,\dots,x_n),(\mu_1,\dots,\mu_n))$}
  	\label{alg:shareData}
  	{\small
  	\begin{enumerate}
  	\item Let $\cY^* = \{\hat{y}|x_1,\dots,x_n\}$ and $\delta^* = \score(\cY_0)$.
	\item Construct a complete weighted bipartite graph $G = (L,R,E)$ where $L = [n], R = [n], E = L \times R$. 
		For each edge $(i,t)$, assign a weight $w(i,t) = \frac{ \alpha_i}{\beta^t} + (n-t)\mu_i$.
	\item Let $M$ be the minimum-weight perfect matching on $G$. 
		For each node $t \in R$, let $\pi(t) \in L$ be the node that it is matched with. 
		If the weight of $M$ is larger than $\delta^*$, output $\tt{NONE}$. Else, define $\delta_{\pi(n)} = \delta^*$, $\cY_{\pi(n)}=\cY^*$. 
	\item For $t$ from $n$ to $2$:
		\begin{itemize}
		\item Let $\delta_{\pi(t-1)} = \delta_{\pi(t)} - \frac{\alpha_{\pi(t)}}{\beta^t} - \sum_{\tau=1}^{t-2} \mu_{\pi(\tau)}$. 
		\item Let $\cY_{\pi(t-1)}$ be such that $\score(\cY_{\pi(t-1)})=\delta_{\pi(t-1)}$.
		\end{itemize}
	\item Output $\omega= (\pi,(\cY_{\pi(1)},\dots,\cY_{\pi(n)}))$.
	\end{enumerate}
	}	
\end{algorithm}

\iffalse
\begin{theorem}\label{thm:shareData}
The algorithm $\ShareData$ can be computed in polynomial time, and outputs a collaborative equilibrium if it exists, and it outputs ${\tt NONE}$ if no such equilibrium exists.
\end{theorem}
\fi
\begin{proof}
The fact that the algorithm runs in polynomial time is immediate, since: 
\begin{itemize}
\item additions, comparisons, and finding minimum weight matchings in a graph \cite{Edmonds} can all be done in (randomized) polynomial time; and
\item the Efficient Output Divisibility Condition implies that computing a distribution $\cY_{\pi(t-1)}$ such that $\score(\cY_{\pi(t-1)})  = \delta_{\pi(t-1)}$ is efficient.
\end{itemize} 

Recall that $(\vec{x},\vec{\lambda})$ supports a collaborative equilibrium if and only if there exists a permutation $\pi$ such that
$$\sum_{t=1}^n \frac{\alpha_{\pi(t)}}{\beta^t} + \sum_{t=1}^n (n-t)\mu_{\pi(t)} \leq \score(\{\hat{y}|x_1,\dots,x_n\})-\score(\cY_0).$$

Note that our algorithm constructs a complete bipartite graph $G= (L\cup R, E)$ where the weight on every edge is $w(i,t) = \frac{ \alpha_i}{\beta^t} + (n-t)\mu_i .$ A matching $M$ on this graph induces a permutation $\pi$ where, for every $t \in R$, we have $\pi(t) = i$ such that $(i,t) \in M$. The weight of such a matching is
$$\sum_{t=1}^n \frac{\alpha_{\pi(t)}}{\beta^t} + \sum_{t=1}^n (n-t)\mu_{\pi(t)}.$$

Thus, $(\vec{x},\vec{\lambda})$ supports a collaborative equilibrium if and only if 
the maximum-weight matching in $G$ has weight less than or equal to $w^*\defeq\score(\{\hat{y}|x_1,\dots,x_n\})-\score(\cY_0)$. 
Note that when the weight of the maximum-matching is greater than $w^*$, 
our algorithm outputs ${\tt NONE}$, indicating that an equilibrium is not supported by the inputs. 

Finally, when the weight of the maximum matching is less than or equal to $w^*$, 
the algorithm outputs a pair $(\pi,\vec{\cY})$ which (by construction) satisfies 
$\score(\cY_{\pi(t-1)})-\score(\cY_{\pi(t)}) =  \frac{\alpha_{\pi(t)}}{\beta^t} + \sum_{\tau=1}^{t-2} \mu_{\pi(\tau)}$ for all $t\in[n]$, 
so the sufficient conditions for $(\pi,\vec{\cY})$ to be a collaborative equilibrium are satisfied. 
\end{proof}

%Finally, we use the above two lemmas to provide a concrete polynomial time algorithm for the mechanism.

%that, 
%given $(\mathbf{x},\mathbf{\alpha},\mathbf{\mu})$, finds a collaborative equilibrium if it exists, 
%and outputs $\tt{NONE}$ if such an equilibrium does not exist. 
%Due to space constraints, the algorithm and proof of correctness are given in Appendix \ref{appx:shareData}.

\subsection{General $\NP$-completeness}\label{sec:NPC}
One may wonder if we can get an efficient mechanism for learning vectors which are not $\rankOne$.
We show that this is unlikely, since
finding a collaborative equilibrium is $\NP$-complete even under a weak generalization of $\rankOne$ learning vectors.
%Due to space, see Appendix \ref{appx:NPC} for details and proof.

%One may wonder if restricting ourselves to $\rankOne$ learning vectors may be unnecessary. We show that this is not probably the case, since finding a collaborative equilibrium is $\NP$-complete even under a weak generalization of $\rankOne$ learning vectors.
\begin{definition}
We say that a learning vector $\lambda \in \Lambda$ is $\rankTwo$ if there exists a non-negative matrix $(\mu_{i,j})_{(i,j)\in[n]\times[n]}$ such that
$\lambda_{\pi,\pi(t)} = \sum_{\tau=1}^{t-1} \mu_{\pi(t),\pi(\tau)}$.
We denote by $\Lambda_2 \subset \Lambda$ the set of all $\rankTwo$ learning vectors. 
\end{definition}

When $\lambda$ is an $\rankTwo$ learning vector, the amount that player $\pi(t)$ learns from  $\pi(\tau)$'s output is bounded above by $\mu_{\pi(t),\pi(\tau)}$. Thus, the total amount that player $\pi(t)$ learns from all prior outputs is $\sum_{\tau=1}^{t-1} \mu_{\pi(t),\pi(\tau)}$. The corresponding necessary condition for a collaborative equilibrium to be supported by some $(\vec{x},\vec{\lambda})$
is that there is a permutation $\pi$ such that
$$\sum_{t=1}^n \frac{\alpha_{\pi(t)}}{\beta^t} + \sum_{t=1}^n \sum_{s > t} \mu_{\pi(s),\pi(t)} \leq \score(\{\hat{y}|x_1,\dots,x_n\})-\score(\cY_0).$$

We show that even checking whether this condition holds is $\NP$-complete.

\begin{theorem}\label{thm:NPC}
Given model parameters $\pparams$, input datasets $(x_1,\dots,x_n)\in X$, and a $\rankTwo$ learning bound vector $(\mu_{i,j})_{(i,j)\in[n]\times[n]}$, 
it is $\NP$-complete to decide whether there exists $\pi$ such that  
$$\sum_{t=1}^n \frac{\alpha_{\pi(t)}}{\beta^t} + \sum_{t=1}^n \sum_{s > t} \mu_{\pi(s),\pi(t)} \leq \score(\{\hat{y}|x_1,\dots,x_n\})-\score(\cY_0).$$
\end{theorem}
\begin{proof}
It is clear that the problem is in $\NP$, since given a permutation $\pi$, the left-hand side can be efficiently computed and compared to the right-hand side of the inequality. 

To show that the problem is $\NP$-hard, we reduce it to the {\em minimum weighted feedback arc set problem}. The unweighted version of this problem was shown to be $\NP$-complete by Karp \cite{Karp}, and the weighted version is also $\NP$-complete \cite{Even}.

\begin{framed}
{\small
$\MinWeightFAS$

\smallskip
\begin{minipage}{0.9\textwidth}{
\vskip5pt
{\sc inputs:} A graph $G = (V,E)$ and a weight function $w:E \to \bR_{\geq 0}$, a threshold $\gamma \in \bR_{\geq 0}$.
\vskip5pt
{\sc output:} Whether or not there exists a set $S \subset E$ of edges which intersects every cycle of $G$ and has weight less than $\gamma$.

%\begin{newitemize}
%\item {\sc inputs:} A graph $G = (V,E)$ and a weight function $w:E \to \bR_{\geq 0}$, a threshold $\gamma \in \bR_{\geq 0}$.
%\item {\sc output:} Whether or not there exists a set $S \subset E$ of edges which intersects every cycle of $G$ and has weight less than $\gamma$.
%\end{newitemize}

}\end{minipage}
}
\end{framed}

All that we need to show is that, given a graph $G$, a set $S$ of edges is a feedback arc set if and only if there exists a permutation $\pi$ of the vertices of $V$ such that $S = \{(\pi(t),\pi(s)) \in E : s < t\}$. 

To see this, note that if $\pi$ is a permutation and $S = \{(\pi(t),\pi(s)) \in E : s < t\}$ then the set $S$ intersects every cycle of $G$. This is because, if 
\ifdoublecolumn
$$C = \{(\pi(i_1),\pi(i_2)),(\pi(i_2),\pi(i_3)),\dots,(\pi(i_k),\pi(i_1))\}$$
\else
$C = \{(\pi(i_1),\pi(i_2)),(\pi(i_2),\pi(i_3)),\dots,(\pi(i_k),\pi(i_1))\}$ 
\fi
is a cycle in $G$, then there must exist $s,t$ such that $s < t$ and $(\pi(t),\pi(s)) \in C$, so $S$ intersects $C$. Thus, $S$ is a feedback arc set.

Conversely, if $S$ is a feedback arc set, then $G' = (V,E-S)$ is a directed acyclic graph, and we can induce an ordering $\pi$ on $V$ following topological sort. Any edge $(\pi(t),\pi(s)) \in E-S$ must satisfy $t < s$. Thus, any edge $(\pi(t),\pi(s))$ where $s < t$ must be in $S$. Thus, given $\pi$ from the topological sort, we must have $S \supset \{(\pi(t),\pi(s)) \in E : s < t\}.$ Since weights are non-negative, the minimal feedback arc set $S^*$ will correspond to a permutation $\pi^*$ such that $S^* = \{(\pi^*(t),\pi^*(s)) \in E : s < t\}$. 

We show how to reduce  $\MinWeightFAS$ to our problem. 
Given $G = (V,E)$, $w:E \to \bR_{\geq 0}$ and $t \in \bR_{\geq 0}$, 
let $\score(\{\hat{y}|x_1,\dots,x_n\})-\score(\cY_0)= \gamma$ and let $\mu_{i,j} = w(i,j)$ if $(i,j) \in E$ and $\mu_{i,j} = 0$ otherwise. 
Suppose there exists a permutation $\pi$ such that 
$$\sum_{t=1}^n \frac{\alpha_{\pi(t)}}{\beta^t} + \sum_{t=1}^n \sum_{s > t} \mu_{\pi(s),\pi(t)} \leq \score(\{\hat{y}|x_1,\dots,x_n\})-\score(\cY_0).$$
Then, since the $\alpha_i$ and $\beta$ are positive,
$$\sum_{t=1}^n \sum_{s > t} \mu_{\pi(s),\pi(t)} \leq \score(\{\hat{y}|x_1,\dots,x_n\})-\score(\cY_0).$$
Plugging in our choices of $\mu_{i,j}$ and $\score(\{\hat{y}|x_1,\dots,x_n\})-\score(\cY_0)$, this becomes
$$\sum_{(\pi(s),\pi(t)) \in E: s > t} w(\pi(s),\pi(t)) \leq \gamma.$$
Since the set $S = \{(\pi(s),\pi(t)) \in E: s > t\}$ is a feedback arc set, we have that there exists a feedback arc set with weight less than $\gamma$.

Conversely, assume no such permutation $\pi$ exists. That is,
$$\sum_{(\pi(s),\pi(t)): s > t} \mu_{\pi(s),\pi(t)} > \score(\{\hat{y}|x_1,\dots,x_n\})-\score(\cY_0)$$ for all permutations $\pi$. 
Note that whether $s$ comes before $t$ or vice-versa does not matter, since this inequality holds for all permutations. Thus, we can also write
$$\sum_{(\pi(s),\pi(t)): s < t} \mu_{\pi(s),\pi(t)} > \score(\{\hat{y}|x_1,\dots,x_n\})-\score(\cY_0)$$
for all permutations $\pi$.
From the argument above, the minimum weight feedback arc set $S^*$ induces a permutation $\pi^*$ such that $S^* = \{(\pi^*(t),\pi^*(s)) \in E : s < t\}$. The weight of $S^*$ is
\ifdoublecolumn
\begin{gather*}
\sum_{(\pi^*(s),\pi^*(t)) \in E: s < t} w(\pi^*(s),\pi^*(t))= \\ \sum_{(\pi^*(s),\pi^*(t)): s < t} \mu_{\pi^*(s),\pi^*(t)} > \score(\{\hat{y}|x_1,\dots,x_n\})-\score(\cY_0) = \gamma.
\end{gather*}
\else
$$\sum_{(\pi^*(s),\pi^*(t)) \in E: s < t} w(\pi^*(s),\pi^*(t))=\sum_{(\pi^*(s),\pi^*(t)): s < t} \mu_{\pi^*(s),\pi^*(t)} > \score(\{\hat{y}|x_1,\dots,x_n\})-\score(\cY_0) = \gamma.$$
\fi
Thus, there does not exist a feedback arc set with weight less than or equal to $\gamma$.

We conclude that if we can efficiently check whether 
$$\sum_{t=1}^n \frac{\alpha_{\pi(t)}}{\beta^t} + \sum_{t=1}^n \sum_{s > t} \mu_{\pi(s),\pi(t)} \leq \score(\{\hat{y}|x_1,\dots,x_n\})-\score(\cY_0),$$ 
then we can efficiently check whether there exists a feedback arc set $S$ with weight less than $\gamma$. Thus, the feedback arc set problem reduces to ours, and our problem is $\NP$-complete.
\end{proof}

We have shown that in our model of scientific collaboration, it can indeed be very beneficial to 
\emph{all parties involved} to collaborate under certain ordering functions,
and such beneficial collaboration outcomes can be efficiently computed under certain realistic conditions (but probably not in the general case).

\section{Ordered MPC}\label{sec:queuedDef}

%In this and later sections we will introduce our definitions of queued MPC, including \emph{ordered} and \emph{time-delayed} variants, 
%and give corresponding secure protocol constructions that are directly applicable to the collaboration scenarios described earlier in the paper.
%We begin by defining the simpler notion of ordered MPC, then we consider timed-delay MPC.

We introduce formal definitions of ordered MPC and associated notions
of fairness and ordered output delivery, and give protocols that realize these notions.
Our definitions build upon the standard security notion\footnote{Note that throughout this work, 
we use ``stand-alone'' security notions rather than ``universally composable'' ones.} 
for traditional MPC, which is described formally in Appendix \ref{appx:mpcSecurity}.

\paragraph{Notation}
For a finite set $A$, we will write $a\larr A$ to denote that $a$ is drawn uniformly at random from $A$. 
For $n\in\NN$, $[n]$ denotes the set $\{1,2,\dots,n\}$. The operation $\oplus$ stands for exclusive-or.
The relation $\compIndist$ denotes computational indistinguishability.
$\negl(n)$ denotes a negligible function in $n$, and $\poly(n)$ denotes a polynomial in $n$.
$\circ$ denotes function composition, and for a function $f$, we write $f^t$ to denote $\underbrace{f\circ f\circ\dots\circ f}_{t}$.

Throughout this work,
we consider computationally bounded (rushing) adversaries in a synchronous complete network, and
we assume the players are honest-but-curious, since any protocol secure in the presence of honest-but-curious players
can be transformed into a protocol secure against malicious players \cite{GMW87}.

\subsection{Definitions}

Let $f$ be an arbitrary $n$-ary function and $p$ be an $n$-ary function that outputs permutation $[n]\rarr[n]$.
An ordered MPC protocol is executed by $n$ parties, where each party $i\in[n]$ has a private input $x_i\in\{0,1\}^*$,
who wish to securely compute
$f(x_1,\dots,x_n) = (y_1,\dots,y_n)\in(\{0,1\}^*)^n$
where $y_i$ is the output of party $i$.
Moreover, the parties are to receive their outputs in a particular \emph{ordering} dictated by
$p(x_1,\dots,x_n) = \pi\in\left([n]\rarr[n]\right)$.
That is, for all $i<j$, party $\pi(i)$ must receive his output \emph{before} party $\pi(j)$ receives her output.
Note that the output ordering $\pi$ is \emph{data-dependent}, as $p$ is a function of the parties' inputs.
%Without loss of generality, we consider $f$ and $p$ to be deterministic, 
%since randomness can be provided by parties as input.

%As in standard multiparty computation, we may consider either a \emph{static} adversary who chooses before the start of the protocol which players to corrupt,
%or an \emph{adaptive} adversary who can perform corruptions at any point during protocol execution, up to some fixed maximum of $t<n$ corruptions.
%The corrupted players may pool their information and deviate from the protocol in arbitrary ways, but we assume that the inputs of all parties are uncorrupted.

Following \cite{GMW87}, the security of ordered MPC with respect to a functionality $f$ and permutation function $p$
is defined by comparing the execution of a protocol
to an ideal process $\IdealFuncOrdered$ where the outputs and ordering are computed by a trusted party who sees all the inputs.
An ordered MPC protocol $F$ is considered to be secure if for any real-world
adversary $\Adv$ attacking the real protocol $F$, there exists an ideal adversary $\Sim$ 
in the ideal process
whose outputs (views) are indistinguishable
from those of $\Adv$.
Note that this implies that no player learns more information about the other players' inputs than can be learned from his own input and output,
\emph{and his own position in the output delivery order}.
The latter condition is important because the output ordering depends on parties' private inputs, 
and thus we require that the protocol reveals as little information as possible about the ordering.

\iffalse
An ordered MPC protocol $F$ is considered to be secure if for any real-world
adversary $\Adv$ attacking the real protocol $F$, there exists an ideal adversary $\Sim$ 
in the ideal process
whose outputs (views) are indistinguishable
from those of $\Adv$.
The ideal functionality $\IdealFuncOrdered$ for ordered MPC and the corresponding security definition
 are detailed formally in Appendix \ref{appx:orderedIdealFunc}.
Note that the standard MPC privacy requirement must be adapted for the ordered MPC setting: we require that
no player learns more information about  the other players' inputs than can be learned from his own input and output,
\emph{and his own position in the output delivery order}.
Moreover, we refine the classical requirements of guaranteed output delivery and fairness as detailed in the next section.
\fi

\paragraph{Many rather than one view}
In the ordered MPC setting, the ideal adversary $\Sim$ and the real-world adversary $\Adv$ each output a view after each output phase.
This is in contrast to standard MPC, where the adversaries simply output one view at the end of the protocol execution.

\begin{idealfunc}{$\IdealFuncOrdered$}
	In the ideal model, a trusted third party $T$ is given the inputs, computes the functions $f,p$ on the inputs, 
	and outputs to each player $i$ his output $y_i$ in the order prescribed by the ordering function.
	In addition, we model an ideal process adversary $\Sim$ who attacks the protocol by corrupting players in the ideal setting.

	\smallskip\noindent
	\textbf{Public parameters.} 
	$\kappa\in\NN$, the security parameter; $n\in\NN$, the number of parties;
	$f:(\{0,1\}^*)^n\rarr(\{0,1\}^*)^n$, the function to compute; and
	$p:(\{0,1\}^*)^n\rarr([n]\rarr[n])$, the ordering function.

	\smallskip\noindent
	\textbf{Private parameters.}
	Each player $i\in[n]$ has input $x_i\in\{0,1\}^*$.

	\begin{enumerate}
	\item
	{\sc Input.}
	Each player $i$ sends his input $x_i$ to $T$.

	\item
	{\sc Computation.}
	$T$ computes $(y_1,\dots,y_n)=f(x_1,\dots,x_n)$ and $\pi=p(x_1,\dots,x_n)$.

	\item
	{\sc Output.}
	The output proceeds in $n$ sequential output rounds. 
	At the start of the $j^{th}$ round, $T$ sends the output value $\outp_{i,j}$ to each party $i$,
	where $\outp_{j,j}=y_{\pi(j)}$ and $\outp_{i,j}=\bot$ for all $i\neq j$.
	When party $\pi(j)$ receives his output, he responds to $T$ with the message $\sf ack$. (The players who receive $\bot$ are not expected to respond.)
	Upon receipt of the $\ack$, $T$ proceeds to the $(j+1)^{th}$ round --
	or, if $j=n$, then the protocol terminates.

	%In the case that $T$ does not receive an $\sf ack$ after sending an output to a (corrupt) player,
	%the protocol terminates (and the rest of the players do not get their outputs).

	\item
	{\sc Output of views.}
	At each output round, after receiving his message from $T$,
	each party produces an output, as follows.
	Each uncorrupted party $i$ outputs $y_i$ if he
	has already received his output, or $\bot$ if he has not. Each corrupted party outputs $\bot$.
	Additionally, the adversary $\Sim$ outputs an arbitrary function of the information that he has learned during the execution of the ideal protocol.

	Let the output of party $i$ in the $j^{th}$ round be denoted by $\view_{i,j}$, 
	and let the view outputted by $\Sim$ in the $j^{th}$ round be denoted by $\view_{\Sim,j}$.
	Let $\IdealViewOrdered$ denote the collection of all views for all output rounds:
	\ifdoublecolumn
	\begin{gather*}
	\IdealViewOrdered= \\ \left((\view_{\Sim,1},\view_{1,1},\dots,\view_{n,1}),\dots,(\view_{\Sim,n},\view_{1,n},\dots,\view_{n,n})\right).
	\end{gather*}
	\else
	\begin{gather*}
	\IdealViewOrdered=\left((\view_{\Sim,1},\view_{1,1},\dots,\view_{n,1}),\dots,(\view_{\Sim,n},\view_{1,n},\dots,\view_{n,n})\right).
	\end{gather*}
	\fi
	(If the protocol is terminated early, then views for rounds which have not yet been started are taken to be $\bot$.)
	\end{enumerate}

\end{idealfunc}

\begin{definition}[Security]\label{def:orderedSecurity}
A multi-party protocol $F$ is said to securely realize $\IdealFuncOrdered$, if the following conditions hold.
\begin{enumerate}
\item The protocol description specifies $n$ check-points 
$C_1,\dots,C_n$ corresponding to events during the execution of the protocol.
\item Take any \PPT{} adversary $\Adv$ who corrupts a subset of players $S\subset[n]$,
and let $V_{\Adv,j}$ be the result of  an arbitrary function $A$ applies to his view  after each check-point $C_j$.
Let $$\RealViewOrdered_\Adv=\left((V_{\Adv,1},V_{1,1},\dots,V_{n,1}),\dots,(V_{\Adv,n},V_{1,n},\dots,V_{n,n})\right)$$
be the tuple consisting of the adversary $\Adv$'s outputted views along with the outputs of the real-world parties 
as specified in the ideal functionality description.
Then there is a \PPT{} ideal adversary $\Sim$
which, attacking $\IdealFuncOrdered$ by corrupting the same subset $S$ of players, 
can output views $\view_{\Sim,j}$ such that
for each $j \in [n]$, it holds that 
$\view_{\Sim,j}\compIndist V_{\Adv,j}$.
\end{enumerate}
\end{definition} 

In the context of ordered MPC, the standard guaranteed output delivery notion is insufficient.
Instead, we define \emph{ordered output delivery}, which requires in addition that all parties receive their
outputs in the order prescribed by $p$.

\begin{definition}[Ordered output delivery]\label{def:orderedOutputDelivery}
An ordered MPC protocol satisfies \emph{ordered output delivery} if for any inputs $x_1,\dots,x_n$,
functionality $f$, and ordering function $p$, it holds that all parties receive their outputs before protocol termination,
and moreover, if $\pi(i)<\pi(j)$, then party $i$ receives his output before party $j$ receives hers,
where $\pi=p(x_1,\dots,x_n)$.
\end{definition}

%Another common requirement in standard MPC is fairness: that is, either all parties receive their output, or none do.
We also define a natural relaxation of the fairness requirement for ordered MPC, called \emph{prefix-fairness}.
Although it is known that fairness is impossible for general functionalities in the presence of a dishonest majority,
we show in the next subsection that prefix-fairness can be achieved even when a majority of parties are corrupt.
We emphasize that this notion relaxes \emph{only} the fairness requirement: 
that is, prefix-fair protocols satisfy full privacy (and correctness) guarantees.

\begin{definition}[Prefix-fairness]\label{def:fairOrdered}
An ordered MPC protocol is \emph{prefix-fair} if for any inputs $x_1,\dots,x_n$,
it holds that the set of parties who have received their outputs at the time of protocol termination (or abortion)
is a prefix of $(\pi(1),\dots,\pi(n))$, where $\pi=p(x_1,\dots,x_n)$ is the permutation induced by the inputs.
\end{definition}

%sPrefix-fairness is a relaxed version of the standard fairness notion in MPC,
%and the standard fairness notion is known to be impossible for general functionalities.
Prefix-fairness can be useful, for example, in settings where it is more important for one party to receive the output than the other;
or where there is some prior knowledge about the trustworthiness of each party (so that more trustworthy parties may receive their outputs first).
%The latter model was initially explored by \cite{ALZ13}. 

\subsection{Construction}

Ordered MPC is achievable by using standard protocols for general MPC,
as described in Protocol~\ref{prot:ordered} below. 
The protocol has $n$ sequential output phases, so that the $n$ outputs can be issued in order.
A subtle point is that because the ordering is a function of the input data, knowledge of the ordering may reveal information
about the input data. Thus, we have to ``mask'' the output values such that each party only learns 
the minimal possible amount of information about the ordering: namely, his own position in the ordering.

\iffalse
\begin{remark}
In our protocols, we assume for simplicity of exposition that the function $f$ to be computed by the MPC outputs a single value $y$
which is outputted to all players. This is without loss of generality: if, instead, $f$ outputs a vector $(y_1,\dots,y_n)$ where party $i$
should receive (only) $y_i$ as output, then we can reduce this to the single-output case by setting $y=(y_1\oplus r_1,\dots,y_n\oplus r_n)$ where 
each $r_i$ is a uniformly random string inputted to the protocol by party $i$. Then, upon receiving $y$ as output, 
each party $i$ will be able to recover his own output $y_i$, and $y$ information-theoretically reveals nothing about any other party's output.
\end{remark}
\fi

\begin{protocol}{Ordered MPC}\label{prot:ordered}
\noindent\textbf{Public parameters.} 
$\kappa\in\NN$, the security parameter; $n\in\NN$, the number of parties; $k\in\NN$, an upper bound on the number of corrupt parties;
$f:(\{0,1\}^*)^n\rarr(\{0,1\}^*)^n$, the function to be computed;
and $p:(\{0,1\})^*\rarr([n]\rarr[n])$, the ordering function.

\begin{enumerate}
\item \label{itm:orderedPrecompute} {\bf Computing shares of $(\pi,\mathbf{y})$:}
Using any general secure MPC protocol (such as \cite{GMW87}) on inputs $x_1,\dots,x_n$,
jointly compute a $k$-out-of-$n$ secret-sharing\footnotemark of $(\pi,\mathbf{y})$ where
$\mathbf{y}=(y_1,\dots,y_n)=f(x_1,\dots,x_n)$ and $\pi=p(x_1,\dots,x_n)$ is a permutation of $[n]$.
At the end of this step, each player possesses a share of the outputs $\mathbf{y}=(y_1,\dots,y_n)$ and of the permutation $\pi$.
\item \label{itm:orderedOutput} {\bf Outputting $y_1,\dots,y_n$ in $n$ phases:}
%The output $y$ will be issued to the players in the order prescribed by permutation $\pi$, as follows.
In the $i^{th}$ output phase, player $\pi^{-1}(i)$ will learn his output.
In phase $i$ the parties run a new instance of a general secure MPC protocol such that:
\begin{itemize}
\item Player $j$'s inputs to the protocol are: the shares of $\mathbf{y}$ and $\pi$ that he got in step \ref{itm:orderedPrecompute}, 
and a random string $r_{i,j}$.
\item The functionality computed is:
\begin{quote}\small\tt
for $j$ from 1 to $n$: if $\pi(j)=i$ then $z_{i,j}:=y_j\oplus r_{i,j}$ else $z_{i,j}=\bot\oplus r_{i,j}$. \\
output $z_i=(z_{i,1},\dots,z_{i,n})$.
\end{quote}
where $\bot$ is a special string that lies outside the output domain. 
\item To recover his output,
each player $j$ computes $y'_{i,j}=z_{i,j}\oplus r_{i,j}$ for all $i$. By construction, 
there is exactly one $i\in[n]$ for which $y'_{i,j}\neq\bot$, and that is equal to the output value $y_j$ for player $j$.
\end{itemize}
%Moreover, in the $i^{th}$ round, all players other than $\pi(i)$ will receive a ``fake'' output with value $\bot$.
%All of the values issued to party $i$ during these output rounds will be masked by a secret random value known only to party $i$,
%so that it is not revealed to anyone but the recipient who receives his ``real'' output in which round.
\end{enumerate}
\medskip\noindent
\textbf{Check-points.}
There are $n$ check-points. 
%For $i\in[n]$, the check-point $C_i$ is the event of $z_i$ (i.e. the output of the $i^{th}$ phase) being learned by all players.
For $i\in[n]$, the check-point $C_i$ is at the end of the $i^{th}$ output phase, when $z_i$ is learned by all players.

\medskip\noindent
\textbf{In case of abort.}
When running the protocol for the honest majority setting, the honest players continue until the end of the protocol regardless of other players' behavior. 
When running the protocol for dishonest majority, if any party aborts in an output phase\footnotemark, then the honest players do not continue to the next phase.
\end{protocol}
\addtocounter{footnote}{-1}
\footnotetext{The standard definition of a secret-sharing scheme can be found in Appendix~\ref{appx:secretSharing}.} 
%For the honest majority setting, we set $k=\lceil n/2\rceil$. For the dishonest majority setting, $k=n$.}
\stepcounter{footnote}
\footnotetext{Each output phase consists of an execution of the underlying general MPC protocol. 
If a party aborts at any time during (and before the end of) the execution of the underlying general MPC protocol, 
this fact will be detected by all honest parties by the end of the phase.}

%We remark that the description of Protocol~\ref{prot:ordered} encompasses two variant protocols, 
%depending on the setting of the secret-sharing threshold parameter $k$:
%one variant is for the honest majority setting, and the other is for the dishonest majority case.

In proving the security of Protocol \ref{prot:ordered}, we refer to the security of modular composition of general protocols 
shown by \cite{Canetti00}, Theorem 5. %We give an informal version of the theorem below, and refer the reader to \cite{Canetti00} for the details.

\iffalse
\begin{informalthm}
Suppose that protocols $\rho_1,\dots,\rho_m$ securely evaluate functions $f_1,\dots,f_m$, respectively,
and a protocol $\pi$ securely evaluates a function $g$ while using subroutine calls for ideal evaluation of $f_1,\dots,f_m$.
Then the protocol $\pi^{\rho_1,\dots,\rho_m}$, derived from protocol $\pi$ by replacing every subroutine call for ideal evaluation
of $f_i$ with an invocation of protocol $\rho_i$, securely evaluates $g$.
\end{informalthm}
\fi

\begin{theorem}\label{thm:orderedProtocol}
Protocol~\ref{prot:ordered} securely realizes $\IdealFuncOrdered$. 
\end{theorem}
\begin{proof}
%The correctness and security of the protocol at each check-point follow directly 
%from the security of modular composition of general protocols (\cite{Canetti00}, Theorem 5). 
Let $\rho_0$ denote the general MPC protocol execution in step \ref{itm:orderedPrecompute}, 
and let $\rho_i$ be the general MPC protocol execution in phase $i$ of step \ref{itm:orderedOutput}, for $i\in[n]$.
For $j\in[n]$, let the protocol $\pi_j$ be the concatenation of the protocols $\rho_0,\dots,\rho_j$.
To prove security at each check-point, it is sufficient to prove that $\pi_j$ satisfies security for all $j\in[n]$:
in other words, that the view outputted by any adversary in the real protocol execution at check-point $j$ can be simulated in the ideal execution.
Finally, for all $j\in[n]$, the security of $\pi_j$ follows directly
from the security of modular composition of general protocols (\cite{Canetti00}, Theorem 5).
\end{proof}

\begin{theorem}\label{thm:prefixFairness}
In the case of honest majority, Protocol \ref{prot:ordered} achieves fairness.
In the dishonest majority setting, prefix-fairness is achieved.
\end{theorem}
\begin{proof}
Fairness holds in the honest majority case, since the honest players complete all output phases, and the shares that the honest players hold
are sufficient to reconstruct each output $y_i$ (recall that the secret-sharing threshold $k$ is $\lceil n/2\rceil$ in the honest majority case).
In the dishonest majority setting, prefix-fairness holds since for all $i\in[n]$,
\emph{all} $n$ shares are required in order to reconstruct the output $y_{\pi(i)}$ in output phase $i$, and
\begin{itemize}
\item if the corrupt parties do not abort during the $i^{th}$ output phase, then by the security of Protocol \ref{prot:ordered},
the output $y_{\pi(i)}$ associated with the $i^{th}$ output phase is delivered correctly to party $i$;
\item if the corrupt parties abort during the $i^{th}$ output phase, then no outputs $y_{\pi(j)}$ for $j>i$ will be learned by any player,
since the honest parties will not execute subsequent output phases.\qedhere
\end{itemize}
\end{proof}

\section{Timed-delay MPC}\label{sec:timedDelayMPC}

In this section, we implementing \emph{time delays} between different players receiving their outputs.
The model is exactly as before, with $n$ players wishing to compute a function $f(x_1,\dots,x_n)$ 
in an ordering prescribed by $p(x_1,\dots,x_n)$ -- except that now, there is an additional requirement of a delay
after each player receives his output and before the next player receives her output.
To realize the timed-delay MPC functionality, we make use of time-lock and time-line puzzles,
which are introduced in Sections \ref{sec:timeLock} and \ref{sec:timeLine}.

\subsection{Ideal functionality with time delays}

We measure time delay in units of computation, rather than seconds of a clock:
that is, rather than making any assumption about global clocks (or synchrony of local clocks)\footnote{A
particular issue that arises when considering a clock-based definition is that it is not clear 
that we can reasonably assume or prove that clocks are in synchrony between the real and ideal world --
but this seems necessary in order to prove security by simulation in the ideal functionality.

We remark that if one is happy to assume the existence of a global clock (or synchrony of local clocks),
then there are other ways to implement timed-delay MPC which sidestep many of the issues inherent in the arguably more realistic model
where clocks may not be perfectly synchronized between different (adversarial) parties.
One example is the ``Bitcoin model'' where the assumption is that the Bitcoin block-chain can serve as a global clock:
in this model, existing protocols such as \cite{BK14} implement some time-delays in MPC, and it seems likely that such protocols can be adapted
to achieve our notion of timed-delay MPC.
}, we measure time by the \emph{evaluations of a particular function} (on random inputs), which we call the \emph{clock function}.

%The designated function by which delays are measured is called the \emph{timing function}.
%The timing function should be viewed as a way to \emph{realize} the ideal functionality in the real world:
%the trusted party in the ideal world simply measures time according to a local clock.

%Note that the ideal functionality $\IdealFuncQueued$ has the unusual feature that the parties and ideal adversary 
%continue to output views even after protocol aborts.

%\input{idealFuncDelayedSimple}
\begin{idealfunc}{$\IdealFuncQueued$}
	In the ideal model, a trusted third party $T$ is given the inputs, computes the functions $f,p$ on the inputs, 
	and outputs to each player $i$ his output $y_i$ in the order prescribed by the ordering function.
	Moreover, $T$ imposes delays between the issuance of one party's output and the next.
	In addition, we model an ideal process adversary $\Sim$ who attacks the protocol by corrupting players in the ideal setting.

	\smallskip\noindent
	\textbf{Public parameters.} 
	$\kappa\in\NN$, the security parameter; $n\in\NN$, the number of parties;
	$f:(\{0,1\}^*)^n\rarr(\{0,1\}^*)^n$, the function to be computed;
	$p:(\{0,1\}^*)^n\rarr([n]\rarr[n])$, the ordering function;
	and $G=G(\kappa)\in\NN$, the number of time-steps between the issuance of one party's output and the next. 

	\smallskip\noindent
	\textbf{Private parameters.}
	Each player $i\in[n]$ has input $x_i\in\{0,1\}^*$.

	\begin{enumerate}
	\item
	{\sc Input.}
	Each player $i$ sends his input $x_i$ to $T$.
	If, instead of sending his input, any player sends the message $\quit$, then the computation is aborted.

	\item
	{\sc Computation.}
	$T$ computes $(y_1,\dots,y_n)=f(x_1,\dots,x_n)$ and $\pi=p(x_1,\dots,x_n)$.

	\item
	{\sc Output.}
	\iffalse
	The output proceeds in $n$ sequential output phases, each of which consists of $G+1$ rounds. 
	The first $G$ rounds of a stage may be considered a challenge-response phase.
	For each stage $j$, for each challenge-response round $k\in[G]$, $T$ sends the message $\chall_{j,k}$ to all players.
	Each player $i$ then sends back a response $\resp_{j,k}$.
	Once at least one response is received, $T$ proceeds to the next challenge (even if some of the other players have not responded yet).
	Once the final challenge $\chall_{j,G}$ has been answered (by at least one player),
	then $T$ sends the $j^{th}$ output, $y_{\pi(j)}$, to party $\pi(j)$.%, and sends $\bot$ to each other player.
	This concludes stage $j$ of the output process.
	\fi
	The output proceeds in $n$ sequential output phases. 
	At each phase $j$, $T$ waits for $G$ time-steps, then sends the $j^{th}$ output, $y_{\pi(j)}$, to party $\pi(j)$.

	%At any point during the output process, if a player sends the message $\quit$ to $T$,
	%then $T$ will ignore that player from that point on.

	%We refer to each challenge-response round as a \emph{time-step}.

	\item
	{\sc Output of views.}
	At the end of each output phase, each party produces an output as follows.
	%Each check-point $\check_j$ corresponds to the event of outputting a value $y_{\pi(j)}$ to party $\pi(j)$ (there are $n$ distinct such events).
	Each uncorrupted party $i$ outputs $y_i$ as his view if he has already received his output, or $\bot$ if he has not. 
	Each corrupted party outputs $\bot$.
	Additionally, the adversary $\Sim$ outputs an arbitrary function of the information that he has learned during the execution of the ideal protocol, 
	after each check-point.
	%The adversary may output his view at any time after the current check-point and before the next one
	%(or, if the current checkpoint is the last one, then he can output the view any time after the checkpoint).

	%Let the vector of all check-points be denoted by $\check=(\check_1,\dots,\check_n)$.
	%For $i\in[n]\cup\{\Sim\}$ and $j\in[G]$, let $\view_{i,j}$ 
	%be the view outputted by $i$ for the check-point $\check_j$.
	Let the output of party $i$ in the $j^{th}$ round be denoted by $\view_{i,j}$,
	and let the view outputted by $\Sim$ in the $j^{th}$ round be denoted by $\view_{\Sim,j}$.
	%Let the view outputted by party $i$ in the $j^{th}$ round be denoted by $\view_{j,i}$, 
	%and let the view outputted by $\Sim$ in the $j^{th}$ round be denoted by $\view_{j,\Sim}$.
	Let $\IdealViewQueued$ denote the collection of all views for all output phases:
	\ifdoublecolumn
	\begin{gather*}
	\IdealViewQueued= \\ \left((\view_{\Sim,1},\view_{1,1},\dots,\view_{n,1}),\dots,(\view_{\Sim,n},\view_{1,n},\dots,\view_{n,n})\right).
	\end{gather*}
	\else
	\begin{gather*}
	\IdealViewQueued=\left((\view_{\Sim,1},\view_{1,1},\dots,\view_{n,1}),\dots,(\view_{\Sim,n},\view_{1,n},\dots,\view_{n,n})\right).
	\end{gather*}
	\fi
	%(If the protocol is terminated early, then views for rounds which have not yet been started are taken to be $\bot$.)
	\end{enumerate}
\end{idealfunc}
\footnotetext{The use of checkpoints
is introduced to capture the views of players and the adversary at intermediate points in protocol execution.
}

For an algorithm $\Adv$, let the run-time\footnote{Run-time is, naturally, measured in ``CPU time''
(i.e. the number of instructions executed in the underlying computational model)
as opposed to real-world ``clock time''.} of $\Adv$ on input ${\sf inp}$ be denoted by $\time_\Adv({\sf inp})$.
If $\Adv$ is probabilistic, the run-time will be a distribution over the random coins of $\Adv$.
Note that the exact run-time of an algorithm will depend on the underlying computational model in which the algorithm is run.
In this work, all algorithms are assumed to be running in the same underlying computational model,
and our definitions and results hold regardless of the specific computational model employed.

\begin{definition}[Security]\label{def:security}
A multi-party protocol $F$ (with parameters $\kappa,n,f,p,G$) 
is said to securely realize $\IdealFuncQueued$, if the following conditions hold.
\begin{enumerate}
\item The protocol description specifies $n$ check-points 
$C_1,\dots,C_n$ corresponding to events during the execution of the protocol.
\item \label{itm:securityDelay} 
There exists a ``clock function'' $g$ such that
between any two consecutive checkpoints $C_i,C_{i+1}$ during an execution of $F$,
any one of the parties (in the real world) must be able to 
locally run $\Omega(G)$ sequential evaluations of $g$ on random inputs.
%sufficient time must elapse (in the real world) for any one of the parties to locally run $\Omega(G)$ sequential evaluations of $g$ on random inputs.
$g$ may also be a \emph{protocol} (involving $n'\leq n$ parties) rather than a function, in which case
we instead require that any subset consisting of $n'$ parties must be able to run $\Omega(G)$ sequential executions of $g$ (on random inputs)
over the communication network being used for the main multi-party protocol $F$.
Then, we say that $F$ is ``clocked by $g$''.
\item \label{itm:simulatability} 
Take any \PPT{} adversary $\Adv$ attacking the protocol $F$ by corrupting a subset of players $S\subset[n]$,
which outputs an arbitrary function $V_{\Adv,j}$ of the information that 
it has learned in the protocol execution after each check-point $C_j$. 
Let $$\RealViewQueued_\Adv=\left((V_{\Adv,1},V_{1,1},\dots,V_{n,1}),\dots,(V_{\Adv,n},V_{1,n},\dots,V_{n,n})\right)$$
be the tuple consisting of the adversary $\Adv$'s outputted views along with the views of the real-world parties 
as specified in the ideal functionality description.
Then there is a \PPT{} ideal adversary $\Sim$
which, attacking $\IdealFuncQueued$ by corrupting the same subset $S$ of players, 
can output views $\view_{\Sim,1},\dots,\view_{\Sim,n}$ (at check-points $C_1,\dots,C_n$ respectively) such that
for each $j\in[n]$, it holds that 
\ifdoublecolumn
\begin{gather*}
|\Pr\left[D(\view_{\Sim,j},\view_{1,j},\dots,\view_{n,j})=1\right] \\
-\Pr\left[D(V_{\Adv,j},V_{1,j},\dots,V_{n,j})=1\right]|\\
\leq\negl(\kappa),
\end{gather*}
\else
$$\left|\Pr\left[D(\view_{\Sim,j},\view_{1,j},\dots,\view_{n,j})=1\right]
-\Pr\left[D(V_{\Adv,j},V_{1,j},\dots,V_{n,j})=1\right]\right|\leq\negl(\kappa),$$
\fi
for any distinguisher $D$ such that
$$\Pr_{\vec{v} \gets \mathcal{V}}[\time_D(\vec{v}) \leq j\cdot\time_{\cal G}()] = 1/\poly(\kappa),$$
when $\mathcal{V}$ is the distribution of views outputted by $\Adv$ or $\Sim$ (that is,
for $\mathcal{V} \in \{(\view_{\Sim,j},\view_{1,j},\dots,\view_{n,j})$, $(V_{\Adv,j},V_{1,j},\dots,V_{n,j})\}$),
and $\cal{G}$ is the algorithm that computes the function $g$ sequentially on $G$ random inputs.
\end{enumerate}
\end{definition}

\subsection{Realizing timed-delay MPC with dummy rounds}

%In this section we consider protocols for securely realizing Timed-Delay MPC for the scientific collaboration scenario discussed in Section~\ref{sec:sharingModel}.
A simple protocol for securely realizing timed-delay MPC is to implement delays
by running $G$ ``dummy rounds'' of communication
in between issuing outputs to different players.

\begin{protocol}{Timed-delay MPC with dummy rounds}\label{prot:simple}
\noindent\textbf{Public parameters.} 
$\kappa\in\NN$, the security parameter; $n\in\NN$, the number of parties;
$f:(\{0,1\}^*)^n\rarr(\{0,1\}^*)^n$, the function to be computed;
$p:(\{0,1\})^*\rarr([n]\rarr[n])$, the ordering function;
and $G=\poly(\kappa)$, the number of time-steps between the issuance of one party's output and the next. 

\begin{enumerate}
\item \label{itm:orderedPrecompute} {\bf Computing shares of $(\pi,\mathbf{y})$:}
Using any general secure MPC protocol (such as \cite{GMW87}),
jointly compute an $k$-out-of-$n$ secret-sharing\footnotemark of $(\pi,\mathbf{y})$ where
$\mathbf{y}=(y_1,\dots,y_n)=f(x_1,\dots,x_n)$ and permutation $\pi=p(x_1,\dots,x_n)$ on the players' inputs.
At the end of this step, each player possesses a share of the outputs $\mathbf{y}=(y_1,\dots,y_n)$ and of the permutation $\pi$.
\item \label{itm:orderedOutput} {\bf Outputting $y_1,\dots,y_n$ in $n$ phases:}
%The output $y$ will be issued to the players in the order prescribed by permutation $\pi$, as follows.
The outputs will occur in $n$ phases: in the $i^{th}$ phase, player $\pi^{-1}(i)$ will learn his output.
In each phase, the players first run $G$ ``dummy rounds'' of communication.
A dummy round is a ``mini-protocol'' defined as follows (let this mini-protocol be denoted by $g_{\sf dum}$):
\begin{itemize}
\item each player initially sends the message $\chall$ to every other player;
\item each player responds to each $\chall$ he receives with a message $\resp$.
\end{itemize}
In each phase, after the dummy rounds have been completed, the parties will run a new instance of a general secure MPC protocol. 
In phase $i$:
\begin{itemize} 
\item Player $j$'s inputs to the protocol are: the shares of $\mathbf{y}$ and $\pi$ that he got in step \ref{itm:orderedPrecompute}, 
and a fresh random string $r_{i,j}$.
\item The functionality computed in each phase $i\in[n]$ is:
\begin{quote}\small\tt
for $j$ from 1 to $n$: if $\pi(j)=i$ then $z_{i,j}:=y_j\oplus r_{i,j}$ else $z_{i,j}=\bot\oplus r_{i,j}$. \\
output $z_i=(z_{i,1},\dots,z_{i,n})$.
\end{quote}
where $\bot$ is a special string that lies outside the output domain. 
\item To recover his output,
each player $j$ computes $y'_{i,j}=z_{i,j}\oplus r_{i,j}$ for all $i$. By construction, 
there is exactly one $i\in[n]$ for which $y'_{i,j}\neq\bot$, and that is equal to the output value $y_j$ for player $j$.
\end{itemize}
%Moreover, in the $i^{th}$ round, all players other than $\pi(i)$ will receive a ``fake'' output with value $\bot$.
%All of the values issued to party $i$ during these output rounds will be masked by a secret random value known only to party $i$,
%so that it is not revealed to anyone but the recipient who receives his ``real'' output in which round.
\end{enumerate}
\medskip\noindent
\textbf{Check-points.}
There are $n$ check-points. 
%For $i\in[n]$, the check-point $C_i$ is the event of $z_i$ (i.e. the output of the $i^{th}$ phase) being learned by all players.
For $i\in[n]$, the check-point $C_i$ is at the end of the $i^{th}$ output phase, when $z_i$ is learned by all players.

\medskip\noindent
\textbf{In case of abort.}
When running the protocol for the honest majority setting, the honest players continue until the end of the protocol regardless of other players' behavior. 
When running the protocol for dishonest majority, if any party aborts in an output phase\footnotemark, then the honest players do not continue to the next phase.
\end{protocol}
\addtocounter{footnote}{-1}
\footnotetext{For the honest majority setting, we set $k=\lceil n/2\rceil$. For the dishonest majority setting, $k=n$.}
\stepcounter{footnote}
\footnotetext{Each output phase consists of an execution of the underlying general MPC protocol preceded by $G$ dummy rounds. 
If a party aborts before the completion of the $G$ dummy rounds, this fact will be detected by all parties in the dummy round in which the abort happens,
because every party is supposed to communicate with every other party in each dummy round.
If a party aborts at any time during (and before the end of) the execution of the underlying general MPC protocol, 
this fact will be detected by all honest parties by the end of the phase.}

\begin{theorem}\label{thm:simpleProtSecurity}
In the presence of honest majority, Protocol~\ref{prot:simple} securely realizes $\IdealFuncQueued$ clocked by $g_{\sf dum}$.
\end{theorem}
\begin{proof}
Let $\Adv$ be any \PPT{} adversary attacking Protocol \ref{prot:simple} by corrupting a subset of players $S\subset[n]$, and
let $$\RealViewQueued_\Adv=\left((V_{\Adv,1},V_{1,1},\dots,V_{n,1}),\dots,(V_{\Adv,n},V_{1,n},\dots,V_{n,n})\right)$$
be the tuple consisting of the adversary $\Adv$'s outputted views along with the views of the real-world parties 
(as specified in the description of $\IdealFuncQueued$).
In order to show that condition~\ref{itm:simulatability} of the security definition (Definition~\ref{def:security}) holds,
we need to show that there is a \PPT{} ideal adversary $\Sim$ which, given access to $\IdealFuncQueued$ and corrupting the same subset $S$ of players, 
can output views $\view_{\Sim,j}$ such that $\RealViewQueued_\Adv \compIndist \IdealViewQueued$.

Recall that the adversary's view can be any function of the inputs of the corrupt parties 
and the messages that the corrupt parties see during the protocol execution. In particular, it is sufficient to show that
there is an ideal adversary $\Sim$ which can output views $\view_{\Sim,j}$ which are indistinguishable from
the transcript of all the messages that the corrupt parties see during the real protocol execution.

Protocol~\ref{prot:simple} consists of sequential executions of the underlying general MPC protocol and the mini-protocol $g_{\sf dum}$.
When the mini-protocol executions are removed from Protocol \ref{prot:simple}, the resulting protocol is identical to Protocol \ref{prot:ordered}.
Hence, by Theorem \ref{thm:orderedProtocol}, 
there is an ideal adversary $\Sim'$ which, given access to $\IdealFuncQueued$ and corrupting the same subset $S$ of players, 
can output views $\view_{\Sim',j}$ such that which are indistinguishable from
the transcript of all the messages that the corrupt parties see \emph{during the $n+1$ executions of the underlying general MPC protocol}
within Protocol \ref{prot:simple}.
The only other messages that are sent in Protocol 2 are the ``dummy'' messages
$\chall$ and $\resp$, which are fixed messages that do not depend on the players' inputs.
In fact, the transcript of an execution of $g_{\sf dum}$ is a deterministic sequence of $\chall$ and $\resp$.
It follows that there exists an ideal adversary $\Sim$ which, by calling $\Sim'$ 
and adding the deterministic transcript corresponding to each execution of $g_{\sf dum}$,
can output views $\view_{\Sim,j}$ which are indistinguishable from 
the transcript of all the messages that the corrupt parties see during the real execution of Protocol \ref{prot:simple}.

Finally, it remains to show that condition~\ref{itm:securityDelay} of the security definition (Definition~\ref{def:security}) is satisfied.
The players are literally running $g$ over the MPC network $G$ times in between issuing outputs,
so it is clear that condition~\ref{itm:securityDelay} holds.
\end{proof}

One downside of the simple solution above is that it requires all (honest) parties to be online 
and communicating until the last player receives his output.
To address this, in Section \ref{sec:protocols} we propose an alternative solution based on \emph{timed-release cryptography}, 
at the cost of an additional assumption that all players have comparable computing speed (within a logarithmic factor).

\subsection{Realizing timed-delay MPC with time-lock puzzles}\label{sec:protocols}

Informally, a time-lock puzzle is a primitive which allows ``locking'' of data, such that it will be released after a pre-specified
time delay, and no earlier. Our next protocol, instead of issuing outputs to players in the clear, gives to each party
his output \emph{locked} into a time-lock puzzle; and in order to enforce the desired ordering, 
the delays required to unlock the puzzles are set to be an increasing sequence.
We first give the definition of time-lock puzzles (in Section \ref{sec:timeLock}) then describe and prove security of
our time-lock-based protocol (in Section \ref{sec:protocolDescriptionTimeLocks}).

\subsubsection{Time-lock puzzles}\label{sec:timeLock}

The delayed release of data in MPC protocols can be closely linked to the problem of ``timed-release crypto'' in general,
which was introduced by \cite{May} and constructed first by \cite{RSW96} with their proposal of \emph{time-lock puzzles}.
We assume time-lock puzzles with a particular structure (that is present in all known implementations): namely,
 the passage of ``time'' will be measured by sequential evaluations of a function ($\TimeStep$).
Unlocking a $t$-step time-lock puzzle can be considered analogous to following a chain of $t$ pointers,
at the end of which there is a special value $x_t$ (e.g. a decryption key) that allows retrieval of the locked data.

\begin{center}
\ifdoublecolumn
\medskip
\begin{tikzpicture}[->,>=stealth',shorten >=1pt,auto,node distance=2cm,
  thick,main node/.style={draw},key node/.style={draw,circle,fill=yellow!20},data node/.style={draw,fill=yellow!20}]
  \node[main node] (1) {$x$};
  \node[main node] (2) [right of=1] {$x_1$};
  \node[main node] (3) [right of=2] {$x_2$};
  \node (4) [right of=3] {$~\dots~$};
  \node[key node] (5) [right of=4] {$x_t$};

  \node[data node,below=0.8cm of 5] (6) {\footnotesize locked data};

  \path[every node/.style={font=\sffamily\scriptsize}]
    (1) edge node [above] {\scriptsize $x_1=f(x)$} (2)
    (2) edge node [above] {\scriptsize $x_2=f(x_1)$} (3)
    (3) edge node [above] {\scriptsize $x_3=f(x_2)$} (4)
    (4) edge node [above] {\scriptsize $x_n=f(x_{n-1})$} (5)
    (5) edge[dotted] node [left] {\scriptsize unlock} (6);
\end{tikzpicture}
\medskip
\else
\medskip
\begin{tikzpicture}[->,>=stealth',shorten >=1pt,auto,node distance=2.7cm,
  thick,main node/.style={draw},key node/.style={draw,circle,fill=yellow!20},data node/.style={draw,fill=yellow!20}]
  \node[main node] (1) {$x$};
  \node[main node] (2) [right of=1] {$x_1$};
  \node[main node] (3) [right of=2] {$x_2$};
  \node (4) [right of=3] {$~\dots~$};
  \node[key node] (5) [right of=4] {$x_t$};

  \node[data node,below=0.8cm of 5] (6) {locked data};

  \path[every node/.style={font=\sffamily\small}]
    (1) edge node [above] {$x_1=f(x)$} (2)
    (2) edge node [above] {$x_2=f(x_1)$} (3)
    (3) edge node [above] {$x_3=f(x_2)$} (4)
    (4) edge node [above] {$x_n=f(x_{n-1})$} (5)
    (5) edge[dotted] node [left] {unlock} (6);
\end{tikzpicture}
\medskip
\fi
\end{center}

\iffalse
This structure of time-lock (or time-line) puzzle happens to lend itself well to our timed-delay protocol constructions;
however, we remark that it also seems a very natural realization of the concept of time-lock,
which has been proposed both in a concrete construction in \cite{RSW96} itself,
and in subsequent works such as \cite{MMV13} which refer to a more black-box notion of ``inherently-sequential'' functions, as we too do.
\fi

\begin{definition}[Time-lock puzzle scheme]
A \emph{time-lock puzzle scheme} is a tuple of \PPT{} algorithms 
\ifdoublecolumn
$$T=(\Lock,\TimeStep,\Unlock)$$
\else
$T=(\Lock,\TimeStep,\Unlock)$ 
\fi
as follows:
\begin{itemize}
\item $\Lock(1^\kappa,d,t)$ takes parameters $\kappa\in\NN$ the security parameter, $d\in\{0,1\}^\ell$ the data to be locked,
and $t\in\NN$ the number of steps needed to unlock the puzzle, and outputs a time-lock puzzle 
$P=(x,t,b,a)\in\{0,1\}^n\times\NN\times\{0,1\}^{n''}\times\{0,1\}^{n'}$ 
where $\ell,n,n',n''=\poly(\kappa)$.
\item $\TimeStep(1^\kappa,x',a')$ takes parameters $\kappa\in\NN$ the security parameter, a bit-string $x'\in\{0,1\}^n$, 
and auxiliary information $a'$, and outputs a bit-string $x''\in\{0,1\}^n$.
\item $\Unlock(1^\kappa,x',b')$ takes parameters $\kappa\in\NN$ the security parameter, 
a bit-string $x'\in\{0,1\}^n$, and auxiliary information $b'\in\{0,1\}^{n'}$,
and outputs some data $d'\in\{0,1\}^\ell$.
\end{itemize}
\end{definition}

%Due to space constraints, we refer the reader to Appendix \ref{appx:timeLine} for details of the correctness
%and security definitions for time-lock puzzles.

To unclutter notation, we will sometimes omit the initial security parameter of these functions (writing e.g. simply $\Lock(d,t)$).
We now define some auxiliary functions.
For a time-lock puzzle scheme $T=(\Lock,\TimeStep,\Unlock)$ and $i\in\NN$, let $\IterateTimeStep^T_i$ denote the following function:
\ifdoublecolumn
\begin{gather*}
\IterateTimeStep^T(i,x,a)= \\ \underbrace{\TimeStep(\TimeStep(\dots (\TimeStep(x,a),a) \dots),a)}_{i}.
\end{gather*}
\else
$$\IterateTimeStep^T(i,x,a)=\underbrace{\TimeStep(\TimeStep(\dots (\TimeStep(x,a),a) \dots),a)}_{i}.$$
\fi
Define $\UnlockFull^T$ to be the following function:
$$\UnlockFull^T((x,t,b,a))=\Unlock(\IterateTimeStep^T(t,x,a),b),$$
that is, the function that should be used to unlock a time-lock puzzle outputted by $\Lock$.

The following definitions formalize correctness and security for time-lock puzzle schemes.

\begin{definition}[Correctness]
A time-lock puzzle scheme $T=(\Lock,\TimeStep,\Unlock)$ is \emph{correct}
if the following holds (where $\kappa$ is the security parameter):
$$\Pr_{(x,t,b,a)\larr\Lock(d,t)}\left[\UnlockFull^T((x,t,b,a))\neq d \right]\leq\negl(\kappa).$$
\end{definition}

\begin{definition}[Security]
Let $T=(\Lock,\TimeStep,\Unlock)$ be a time-lock puzzle scheme. 
%(Note that this is a distribution over the random coins of $\Adv$.)
$T$ is \emph{secure} if it holds that:
for all $d,d'\in\{0,1\}^\ell, t=\poly(\kappa)$, if there exists an adversary $\Adv$ 
that solves the time-lock puzzle $\Lock(d,t)$, that is,
$$\Pr_{P\gets \Lock(d,t)}[\Adv(P)=d]=\eps\mbox{ for some non-negligible }\eps,$$
then for each $j\in[t]$, there exists an adversary $\Adv_j$ such that
$$\Pr_{P'\gets \Lock(d',j)}\left[\Adv_j(P')=d'\right] \geq1-\negl(\kappa),\mbox{ and}$$
$$\Pr_{\substack{P \gets \Lock(d,t),\\ P'\gets \Lock(d',j)}}\left[\time_\Adv(P) \geq (t/j) \cdot \time_{\Adv_j}(P')~|~\Adv(P)=d\right] \geq 1-\negl(\kappa).$$
\end{definition}

%\subsection{Time-lock-based construction of timed-delay MPC}\label{sec:protocols}
\subsubsection{Protocol based on time-lock puzzles}\label{sec:protocolDescriptionTimeLocks}

Because of the use of time-lock puzzles by different parties in the protocol that follows,
we require an additional assumption 
that all players have comparable computing power (within a logarithmic factor).

%\begin{assumption}\label{assume:equalPower}
%The difference in speed of performing computations between any two parties $i,j\in[n]$ is at most a factor of $B=O(\log(\kappa))$.
%\end{assumption}

\paragraph{Relative-Delay Assumption.}
The difference in speed of performing computations between any two parties $i,j\in[n]$ is at most a factor of $B=O(\log(\kappa))$.

\begin{protocol}{Timed-delay MPC with time-lock puzzles}\label{prot:timeLock}%REAL
\noindent\textbf{Public parameters.} 
$\kappa\in\NN$, the security parameter; $n\in\NN$, the number of parties;
$f:(\{0,1\}^*)^n\rarr(\{0,1\}^*)^n$, the function to be computed;
$p:(\{0,1\})^*\rarr([n]\rarr[n])$, the ordering function;
$B=O(\log(\kappa))$, the maximum factor of difference between any two parties' computing power;
$G=\poly(\kappa)$, the number of time-steps between the issuance of one party's output and the next;
%and $\cT=\{(\Lock_m,\TimeStep_m,\Unlock_m)\}_{m\in\NN}$ a time-line puzzle scheme.
and $T=\{\Lock,\TimeStep,\Unlock\}$ a time-lock puzzle scheme.

\medskip\noindent
{\sc Inputs.} 
Each party $i$ has input $x_i$.

\medskip\noindent
{\sc Protocol steps.}
Let $(y_1,\dots,y_n)=f(x_1,\dots,x_n)$ and $\pi=p(x_1,\dots,x_n)$.
Define $t_1=1$ and $t_{i+1}=(B\cdot G+1)\cdot t_{i}$ for $i\in[n-1]$. 
Compute $(P_1,\dots,P_n)$, where
each $P_i=(x_i,t_{\pi(i)},a_i,b_i)$ is a time-lock puzzle computed as
$$P_i=\Lock(y_i\oplus r_i,t_{\pi(i)}),$$
where each $r_i$ is a random string provided as input randomness by party $i$.

\medskip\noindent
{\sc Outputs.} 
For each $i\in[n]$, the puzzle $P_i$ is outputted to party $i$.
The players all receive their respective outputs at the same time, 
then recovers his output $y_i$ by solving his time-lock puzzle, and finally ``unmasking'' the result by XORing with his random input $r_i$.

\medskip\noindent
\textbf{Check-points.}
There are $n$ check-points. For $i\in[n]$, the check-point $C_i$ is the event of party $\pi(i)$ learning his eventual output $y_{\pi(i)}$
(i.e. when he finishes solving his time-lock puzzle).
%Denote by $C=(C_1,\dots,C_n)$ the vector of all check-points.

%\medskip\noindent
%\emph{In case of abort.}
%When running the protocol for the honest majority setting, the honest players continue until the end of the protocol regardless of other players' behavior. 
%When running the protocol for dishonest majority, if any party aborts in an output phase\footnotemark, then the honest players do not continue to the next phase.
\end{protocol}

For the following theorem, 
we assume that each player $i$ uses the optimal algorithm to solve his puzzle $P_i$ 
that outputs the correct answer.
Without this assumption, any further protocol analysis would not make sense: there can always be a ``lazy'' player who willfully uses
a very slow algorithm to solve his puzzle, who will as a result learn his eventual output much later in the order than he could otherwise have done.
%(An even lazier player could simply refuse to solve his puzzle, and never learn his eventual output.)
The property that we aim to achieve is that every player \emph{could} learn his output at his assigned position in the ordering $\pi$,
with appropriate delays before and after he learns his output.

\begin{theorem}\label{thm:timeLockProtocol}
Suppose that the Relative-Delay Assumption holds, 
and each player $i$ uses the optimal algorithm to solve his puzzle $P_i$ 
that outputs (with overwhelming probability) the correct answer.
Then, Protocol~\ref{prot:timeLock} securely realizes $\IdealFuncQueued$ 
when there is an honest majority.
\end{theorem}
\begin{proof}
First, we prove that condition~\ref{itm:securityDelay} of the security definition (Definition~\ref{def:security}) is satisfied.
Let $\Adv_i$ denote the algorithm that party $i$ uses to solve his time-lock puzzle,
and let the time at which party $i$ learns his answer $y_i$ be denoted by $\tau_i = \time_{\Adv_i}(P_i)$.
By the security of the time-lock puzzles, there exists an algorithm $\Adv'_i$ that player $i$ could use to solve the puzzle $\Lock(0^\ell,1)$ in time $\tau_i/t_i$.
Moreover, by the Relative-Delay Assumption, it holds that no player can solve the puzzle $\Lock(0^\ell,1)$ more than $B$ times faster than another player:
that is, $\max_i(\tau_i/t_i)\leq B\cdot \min_i(\tau_i/t_i)$.
It follows that even the slowest player (call him $i^*$) would be able to run $t_i/B$ executions of $\Adv'_{i^*}$ within time $\tau_i$, for any $i$.

Without loss of generality, assume that the ordering function $p$ is the identity function.
Consider any consecutive pair of checkpoints $C_i,C_{i+1}$. These checkpoints occur at times $\tau_i$ and $\tau_{i+1}$, by definition.
We have established that in time $\tau_{i}$, player $i^*$ can run $t_{i}/B$ executions of $\Adv'_{i^*}$,
and in time $\tau_{i+1}$, he can run $t_{i+1}/B$ executions of $\Adv'_{i^*}$.
It follows that in between the two checkpoints (i.e. in time $\tau_{i+1}-\tau_i$), he can run $(t_{i+1}-t_{i})/B$ executions of $\Adv'_{i^*}$.
Substituting in the equation $t_{i+1}=(B\cdot G+1)\cdot t_i$ from the protocol definition, 
we get that player $i^*$ can run $G\cdot t_i$ executions of $\Adv'_{i^*}$ between checkpoints $C_i$ and $C_{i+1}$.
Since $t_i\geq1$ for all $i$, this means that $i^*$ can run at least $G$ executions of $\Adv'_{i^*}$ between \emph{any} consecutive pair of checkpoints.
Hence, condition~\ref{itm:securityDelay} holds.

We now prove condition \ref{itm:simulatability}.
Let $\cal G$ be the algorithm that evaluates $\Adv'_{i^*}$ sequentially $G$ times on random inputs.
It is sufficient to show that for any adversary $\Adv$ attacking the protocol by corrupting a subset $S\subset[n]$ of players,
which outputs a view $V_{\Adv,j}$ at each checkpoint $j$ which is the \emph{transcript of all messages that it has seen so far},
there is an ideal adversary $\Sim$ which outputs views $\view_{\Sim,1},\dots,\view_{\Sim,n}$ such that for any $j\in[n]$, 
for any distinguisher $D$ whose run-time satisfies the conditions in Definition \ref{def:security}, item \ref{itm:simulatability},
\ifdoublecolumn
\begin{gather*}
\left|\Pr\left[D(\view_{\Sim,j},\view_{1,j},\dots,\view_{n,j})=1\right]
-\Pr\left[D(V_{\Adv,j},V_{1,j},\dots,V_{n,j})=1\right]\right| \\ \leq\negl(\kappa).
\end{gather*}
\else
$$\left|\Pr\left[D(\view_{\Sim,j},\view_{1,j},\dots,\view_{n,j})=1\right]
-\Pr\left[D(V_{\Adv,j},V_{1,j},\dots,V_{n,j})=1\right]\right|\leq\negl(\kappa).$$
\fi

Recall that there are $n$ sequential output stages in the ideal functionality $\IdealFuncQueued$.
Consider an ideal adversary $\Sim$ attacking $\IdealFuncQueued$ by corrupt a set of parties $S\subset[n]$.
Let $\vec{inp}$ denote the vector of inputs and input randomness of the corrupt parties (note that these are known to $\Sim$).
Take any $i\in[n]$.
In the ideal protocol execution, $\Sim$ learns the following in the $\pi(i)^{th}$ output stage:
\begin{itemize}
\item nothing, if $i\notin S$; or
\item the input value $x_i$, the input randomness $r_i$, and the eventual output $y_{i}$ if $i\in S$.
\end{itemize}
Note that as a result, $\Sim$ learns $\pi(i)$ at output stage $\pi(i)$, for each $i\in S$.
The delay values $t_1,\dots,t_n$ are a fixed sequence of values independent of the parties' inputs, so they are known to $\Sim$.
Thus, at each check-point $j\in[n]$, the ideal adversary $\Sim$ can compute $n$ time-lock puzzles
$$
\hat{P}_{j,i}=\begin{cases}
\Lock(y_i\oplus r_i,t_{\pi(i)}) & \mbox{ if } 1\leq i\leq j \\
\Lock(r_i,t_{\pi(i)}) & \mbox{ if } j<i\leq n
\end{cases}.
$$
%where each $r'_{j,i}$ is chosen uniformly at random by $\Sim$.
Let the ideal adversary $\Sim$ output the following view at each check-point $j$:
$$\view_{\Sim,j}=\Sim_j(\vec{inp},(\hat{P}_{j,1},\dots,\hat{P}_{j,n})),$$
where $\Sim_j$ is the ideal adversary (for the underlying general MPC protocol) that simulates the adversary's $j^{th}$ view $V_{\Adv,j}$.

We now analyze the distribution of the puzzles $\hat{P}_{j,i}$.
For the range $1\leq i\leq j$, the puzzle $\hat{P}_{j,i}$ is by definition identically distributed to 
the puzzle that is outputted to player $i$ in the real execution of Protocol \ref{prot:timeLock}.
Now take any $j\in[n]$, and let $D$ be any distinguisher whose run-time satisfies the conditions in Definition \ref{def:security}, item \ref{itm:simulatability}.
Recall that the players are assumed to solve the time-lock puzzles using the optimal algorithm.
Hence, it follows from the security of the underlying time-lock puzzle scheme 
that for any $i$ in the range $j<i\leq n$,
$$\left|\Pr\left[D(P_{j,i})=1\vphantom{\hat{P}}\right]-\Pr\left[D(\hat{P}_{j,i})=1\right]\right|\leq\negl(\kappa).$$
Since we defined the outputs of $\Sim$ to be 
\ifdoublecolumn
$$\view_{\Sim,j}=\Sim_j(\vec{inp},(\hat{P}_{j,1},\dots,\hat{P}_{j,n}))$$
\else
$\view_{\Sim,j}=\Sim_j(\vec{inp},(\hat{P}_{j,1},\dots,\hat{P}_{j,n}))$
\fi
 for $j\in[n]$, it follows that
\ifdoublecolumn
\begin{gather*}
\left|\Pr\left[D(\view_{\Sim,j},\view_{1,j},\dots,\view_{n,j})=1\right]
-\Pr\left[D(V_{\Adv,j},V_{1,j},\dots,V_{n,j})=1\right]\right| \\ \leq\negl(\kappa)
\end{gather*}
\else
$$\left|\Pr\left[D(\view_{\Sim,j},\view_{1,j},\dots,\view_{n,j})=1\right]
-\Pr\left[D(V_{\Adv,j},V_{1,j},\dots,V_{n,j})=1\right]\right|\leq\negl(\kappa)$$
\fi
as required.
We conclude that Protocol~\ref{prot:timeLock} securely realizes $\IdealFuncQueued$ clocked by $\Adv'_{i^*}$.
\end{proof}

A few remarks are in order.
In Protocol~\ref{prot:timeLock}, 
all the parties can stop interacting as soon as all the puzzles are outputted.
When the locking algorithm $\Lock(d,t)$ has run-time that is  independent of the delay $t$,
the run-time of Protocol~\ref{prot:timeLock} is also independent of the delay parameters.
(This is achievable using the \cite{RSW96} time-lock construction, for example.)
Alternatively, using a single time-line puzzle in place of the time-lock puzzles in Protocol~\ref{prot:timeLock}
can improve efficiency, since the time required to generate a time-line puzzle is dependent only on the longest delay $t_n$,
whereas the time required to generate $n$ separate time-lock puzzles depends on the sum of all the delays, $t_1+\dots+t_n$.

\subsection{Time-line puzzles}\label{sec:timeLine}
%\paragraph{Time-line puzzles}
We now introduce the more general, novel definition of \emph{time-line} puzzles,
which can be useful for locking together many data items with different delays for a single recipient, or for locking
data for a group of people. 
In the latter case, it becomes a concern that computation speed will vary between parties: 
indeed, the scheme will be unworkable if some parties have orders of magnitude more computing power than others,
so some assumption is required on the similarity of computing power among parties, 
such as the Relative-Delay Assumption of Section \ref{sec:protocolDescriptionTimeLocks}.
When a time-line puzzle is given to a single recipient, then no additional assumptions are required.

We remark that time-line puzzles could be used (instead of a set of time-lock puzzles) to realize Protocol \ref{prot:timeLock}.
More generally, we present this new notion because we believe that time-line puzzles may be of independent interest as a timed-release primitive.

In some ways, a time-line puzzle can be thought of as a primitive that packages a sequence of time-lock puzzles together
into a unified system about which we can reason and give security guarantees.
However, time-line puzzles can also provide concrete advantages over a collection of time-lock puzzles.
For example, when issuing many time-lock puzzles to one recipient, the recipient has to run the computation for all of the puzzles in parallel:
that is, he does $O(m\cdot t)$ computation where $m$ is the number of data items and $t$ is the time-delay.
If instead he gets a time-line puzzle, he only has to run one puzzle's worth of computation in order to unlock all the data items:
that is, he does only $O(t)$ computation, just like for a single time-lock puzzle.

\begin{definition}[Time-line puzzles]\label{def:systemTimeLocks}
A \emph{time-line puzzle scheme} is a family of \PPT{} algorithms 
\ifdoublecolumn
$$\cT=\{(\Lock_m,\TimeStep_m,\Unlock_m)\}_{m\in\NN}$$
\else
$\cT=\{(\Lock_m,\TimeStep_m,\Unlock_m)\}_{m\in\NN}$ 
\fi
as follows:
\begin{itemize}
\item $\Lock_m(1^\kappa,(d_1,\dots,d_m),(t_1,\dots,t_m))$ takes parameters $\kappa\in\NN$ the security parameter, 
$(d_1,\dots,d_m)\in(\{0,1\}^\ell)^m$ the data items to be locked,
and $(t_1,\dots,t_m)\in\NN^m$ the number of steps needed to unlock each data item (respectively),
and outputs a puzzle 
\ifdoublecolumn
\begin{gather*}
P=(x,(t_1,\dots,t_m),(b_1,\dots,b_m),a) \\ \in\{0,1\}^n\times\NN\times(\{0,1\}^{n''})^m\times\{0,1\}^{n'}
\end{gather*}
\else
$$P=(x,(t_1,\dots,t_m),(b_1,\dots,b_m),a)\in\{0,1\}^n\times\NN\times(\{0,1\}^{n''})^m\times\{0,1\}^{n'}$$ 
\fi
where $n,n',n''=\poly(\kappa)$, and $a$ can be thought of as auxiliary information.
\item $\TimeStep_m(1^\kappa,x',a')$ takes parameters $\kappa\in\NN$ the security parameter, 
a bit-string $x'\in\{0,1\}^n$,
and auxiliary information $a'$, and outputs a bit-string $x''\in\{0,1\}^n$.
\item $\Unlock_m(1^\kappa,x',b')$ takes parameters $\kappa\in\NN$ the security parameter, 
a bit-string $x'\in\{0,1\}^n$, and auxiliary information $b'\in\{0,1\}^{n'}$,
and outputs some data $d'\in\{0,1\}^\ell$.
\end{itemize}
\end{definition}

In terms of the ``pointer chain'' analogy above, solving a time-line puzzle may be thought of as following a
pointer chain where not one but many keys are placed along the chain, at different locations $t_1,\dots,t_m$.
Each key $x_{t_i}$ in the pointer chain depicted below enables the ``unlocking'' of the locked data $b_i$:
for example, $b_i$ could be the encryption of the $i^{th}$ data item $d_i$ under the key $x_{t_i}$.
%If there are multiple parties involved, it may be that certain keys are only useful to certain parties.

\begin{center}
\ifdoublecolumn
\medskip
\else
\medskip
\begin{tikzpicture}[->,>=stealth',shorten >=1pt,auto,node distance=1.4cm,
  thick,main node/.style={draw},key node/.style={draw,circle,fill=yellow!20},data node/.style={draw,fill=yellow!20}]
  \node[main node] (1) {$x$};
  \node[main node] (2) [right of=1] {$x_1$};
  \node (3) [right of=2] {$~\dots~$};
  \node[main node] (4) [right of=3] {$x_{t_1-1}$};
  \node[key node] (5) [right of=4] {$x_{t_1}$};
  \node[main node] (6) [right of=5] {$x_{t_1+1}$};
  \node (7) [right of=6] {$~\dots~$};
  \node[key node] (8) [right of=7] {$x_{t_2}$};
  \node (9) [right of=8] {$~\dots~$};
  \node[main node] (10) [right of=9] {$x_{t_m-1}$};
  %\node[main node] (10) [right of=7] {$x_{t_m-1}$};
  \node[key node] (11) [right of=10] {$x_{t_m}$};

  \node[data node] (12) [below of=5] {$b_1$};
  \node[data node] (13) [below of=8] {$b_2$};
  \node[data node] (14) [below of=11] {$b_m$};

  \path[every node/.style={font=\sffamily\small}]
    (1) edge node [above] {} (2)
    (2) edge node [above] {} (3)
    (3) edge node [above] {} (4)
    (4) edge node [above] {} (5)
    (5) edge node [above] {} (6)
    (6) edge node [above] {} (7)
    %(7) edge node [above] {} (10)
    (7) edge node [above] {} (8)
    (8) edge node [above] {} (9)
    (9) edge node [above] {} (10)
    (10) edge node [above] {} (11);
  \path[every node/.style={font=\sffamily\small}]
    (5) edge[dotted] node [left] {unlock} (12)
    (8) edge[dotted] node [left] {unlock} (13)
    (11) edge[dotted] node [left] {unlock} (14);
\end{tikzpicture}
\medskip
\fi
\end{center}

Using similar notation to that defined for time-lock puzzles: for a time-line puzzle scheme $\cT$,
let $\IterateTimeStep^{\cT}_{m}$ denote the following function:
\ifdoublecolumn
\begin{gather*}
\IterateTimeStep^{\cT}_{m}(i,x,a)= \\ \underbrace{\TimeStep_m(\TimeStep_m(\dots (\TimeStep_m(x,a),a) \dots),a)}_{i}.
\end{gather*}
\else
$$\IterateTimeStep^{\cT}_{m}(i,x,a)=\underbrace{\TimeStep_m(\TimeStep_m(\dots (\TimeStep_m(x,a),a) \dots),a)}_{i}.$$
\fi
Define $\UnlockFull^{\cT}_{m,i}$ to be the following function:
\ifdoublecolumn
\begin{gather*}
\UnlockFull^{\cT}_{m,i}((x,t_i,b_i,a))= \\ \Unlock_m(\IterateTimeStep^{\cT}_{m}(t_i,x,a),b_i),
\end{gather*}
\else
$$\UnlockFull^{\cT}_{m,i}((x,t_i,b_i,a))=\Unlock_m(\IterateTimeStep^{\cT}_{m}(t_i,x,a),b_i),$$
\fi
that is, the function that should be used to unlock the $i^{th}$ piece of data locked by a time-line puzzle which was generated by $\Lock_m$.
We now define correctness and security for time-line puzzle schemes.

\begin{definition}[Correctness]
A time-line puzzle scheme $\cT$ is \emph{correct} if for all $m=\poly(\kappa)$
and for all $i\in[m]$, it holds that
$$\Pr_{(x,\vec{t},\vec{b},a)\larr\Lock_m(\vec{d},\vec{t})}\left[\UnlockFull^\cT_{i}((x,t_i,b_i,a))\neq d_i \right]\leq\negl(\kappa),$$
where $\kappa$ is the security parameter, $\vec{d}=(d_1,\dots,d_m)$, and $\vec{t}=(t_1,\dots,t_m)$.
\end{definition}

%In the below security definition, for a vector $\vec{t}$, we write $\vec{t'}=(\vec{t}_{-i},z)$ to denote the vector which is identical to $\vec{t}$, except that
%the $i^{th}$ component of $\vec{t}$ is replaced by $z$ in $\vec{t'}$.
	
Security for time-line puzzles involves more stringent requirements than security for time-lock puzzles.
We define security in terms of two properties which must be satisfied: \emph{timing} and \emph{hiding}.
The timing property is very similar to the security requirement for time-lock puzzles,
and gives a guarantee about the relative amounts of time required to solve different time-lock puzzles.
The hiding property ensures (informally speaking) that the ability to unlock any given data item that is
locked in a time-line puzzle does not imply the ability to unlock any others. 
The security definition (Definition~\ref{def:timeLineSec}, below) refers to the following security experiment.

\begin{framed}
\begin{center}
The experiment $\HidingExp_{\Adv,\cT}(\kappa)$
{\small \begin{enumerate}
	%\item The challenger generates a time-line puzzle $(pk,sk)\larr\PGen(1^k)$, and sends $(1^k,pk)$ to $\Adv$.
	\item $\Adv$ outputs $m=\poly(\kappa)$ and data vectors $\vec{d}_0,\vec{d}_1\in(\{0,1\}^\ell)^m$ and a time-delay vector $\vec{t}\in\NN^m$.
	\item The challenger samples $(\beta_1,\dots,\beta_m)\larr \{0,1\}^m$, computes the time-line puzzle 
			$(x,\vec{t}, \vec{b}, a)=\Lock_m(1^\kappa,((d_{\beta_1})_1,\dots,(d_{\beta_m})_m), \vec{t})$, and sends $(x,a)$ to $\Adv$.
	\item $\Adv$ sends a query $i\in[m]$ to the challenger. The challenger responds by sending $b_i$ to $\Adv$. 
			This step may be repeated up to $m-1$ times. Let $I$ denote the set of queries made by $\Adv$.
	\item $\Adv$ outputs $i'\in[m]$ and $\beta'\in\{0,1\}$.
	\item The output of the experiment is $1$ if $i'\notin I$ and $\beta' = \beta_{i'}$. Otherwise, the output is $0$.
\end{enumerate} }
\end{center}
\end{framed}

\begin{definition}[Security]\label{def:timeLineSec}
Let 
\ifdoublecolumn
$$\cT=\{(\Lock_m,\TimeStep_m,\Unlock_m)\}_{m\in\NN}$$
\else
$\cT=\{(\Lock_m,\TimeStep_m,\Unlock_m)\}_{m\in\NN}$ 
\fi
be a time-line puzzle scheme. 
$T$ is \emph{secure} if it satisfies the following two properties.

\begin{itemize}
\item {\sc Timing:}
For all $m=\poly(\kappa)$ and $\vec{d},\vec{d}'\in(\{0,1\}^\ell)^m$ and $\vec{t}=(t_1,\dots,t_m)$, 
if there exists an adversary $\Adv$ that solves any one of the puzzles defined by the time-line, that is,
\ifdoublecolumn
$$\Pr_{P \gets \Lock_m(\vec{d},\vec{t})}[\Adv(P)=d_i]=\eps$$
for some non-negligible $\eps$ and some $i\in[m]$,
\else
$$\Pr_{P \gets \Lock_m(\vec{d},\vec{t})}[\Adv(P)=d_i]=\eps\mbox{ for some non-negligible }\eps\mbox{ and some }i\in[m],$$
\fi
then for all $j\in[t_i]$ and all $\vec{t'}\in[t_m]^m$, there exists an adversary $\Adv_{j,\vec{t'}}$ such that
$$\Pr_{P' \gets \Lock_m(\vec{d}',\vec{t'})}[\Adv_{j,\vec{t'}}(P')=d_j] \geq1-\negl(\kappa),\mbox{ and}$$
%$$\Pr[\Adv'(\Lock_m(\vec{d},(\vec{t}_{-i},1)))=d_i]\geq\eps'\mbox{ for some non-negligible }\eps'\in(0,1],\mbox{ and}$$
\ifdoublecolumn
\begin{align*}
\Pr_{\substack{P \gets \Lock_m(\vec{d},\vec{t}) \\ P' \gets \Lock_m(\vec{d}',\vec{t}')}}[ & \time_\Adv(P) \geq (t'_j/t_i) \cdot \time_{\Adv_{j,\vec{t'}}}(P')\\ ~|~ & \Adv(\Lock_m(\vec{d},\vec{t}))=d_i] \geq 1-\negl(\kappa).
\end{align*}
\else
$$\Pr_{\substack{P \gets \Lock_m(\vec{d},\vec{t}) \\ P' \gets \Lock_m(\vec{d}',\vec{t}')}}[\time_\Adv(P) \geq (t'_j/t_i) \cdot \time_{\Adv_{j,\vec{t'}}}(P')~|~\Adv(\Lock_m(\vec{d},\vec{t}))=d_i] \geq 1-\negl(\kappa).$$
\fi
\item {\sc Hiding:}
For all \PPT{} adversaries $\Adv$, it holds that
$$\Pr[\HidingExp_{\Adv,\cT}(\kappa) = 1] \leq 1/2 + \negl(\kappa).$$
\end{itemize}
\end{definition}

In Appendix \ref{appx:timeLine}, we describe and prove the security of two constructions of time-line puzzle schemes.
One of these schemes is based on a concrete assumption (specifically, on the sequentiality of modular exponentiation, 
like the time-lock puzzles of \cite{RSW96}), whereas the other is based on the existence of a 
``black-box'' inherently-sequential hash function.

\section*{Acknowledgements}

We would like to thank Yehuda Lindell for an interesting discussion on the nature of fairness in multiparty computation,
and we are grateful to Juan Garay, Bj\"{o}rn Tackmann, and Vassilis Zikas for an illuminating discussion about measures of partial fairness in MPC.
We thank Silvio Micali and Ron Rivest for helpful comments about the data-sharing model.

\printbibliography

\begin{appendix}

\section{MPC security definition}\label{appx:mpcSecurity}

\begin{idealfunc}{$\IdealFuncMPC$}
	In the ideal model, a trusted third party $T$ is given the inputs, computes the function $f$ on the inputs, 
	and outputs to each player $i$ his output $y_i$.
	In addition, we model an ideal process adversary $\Sim$ who attacks the protocol by corrupting players in the ideal setting.

	\medskip\noindent
	\textbf{Public parameters.} 
	$\kappa\in\NN$, the security parameter; $n\in\NN$, the number of parties; and
	$f:(\{0,1\}^*)^n\rarr(\{0,1\}^*)^n$, the function to be computed.

	\medskip\noindent
	\textbf{Private parameters.}
	Each player $i\in[n]$ holds a private input $x_i\in\{0,1\}^*$.

	\begin{enumerate}
	\item
	{\sc Input.}
	Each player $i$ sends his input $x_i$ to $T$.

	\item
	{\sc Computation.}
	$T$ computes $(y_1,\dots,y_n)=f(x_1,\dots,x_n)$.

	\item
	{\sc Output.}
	For each $i\in[n]$, $T$ sends the output value $y_i$ to party $i$.

	\item
	{\sc Output of views.}
	After the protocol terminates,
	each party produces an output, as follows.
	Each uncorrupted party $i$ outputs $y_i$ if he
	has received his output, or $\bot$ if not. Each corrupted party outputs $\bot$.
	Additionally, the adversary $\Sim$ outputs an arbitrary function of the information that he has learned during the execution of the ideal protocol.

	Let the output of party $i$ be denoted by $\view_{i}$, 
	and let the view outputted by $\Sim$ be denoted by $\view_{\Sim}$.
	Let $\IdealViewMPC$ denote the collection of all the views:
	$$\IdealViewMPC=\left(\view_{\Sim},\view_{1},\dots,\view_{n}\right).$$
	\end{enumerate}
\end{idealfunc}

\begin{definition}[Security]\label{def:mpcSecurity}
A multi-party protocol $F$ is said to securely realize $\IdealFuncMPC$ if for any \PPT{} adversary $\Adv$ attacking the protocol $F$ by corrupting a subset of players $S\subset[n]$,
there is a \PPT{} ideal adversary $\Sim$
which, attacking $\IdealFuncMPC$ by corrupting the same subset $S$ of players, 
can output a view $\view_{\Sim}$ such that
$$V_\Adv \compIndist \view_{\Sim},$$
where $V_{\Adv}$ is the view outputted by the real-world adversary $\Adv$
(this may be an arbitrary function of the information that $\Adv$ learned in the protocol execution).
%is an arbitrary function $V_{\Adv}$ of the information that 
%it has learned in the protocol execution.
%$F$ is said to securely realize $\IdealFuncMPC$ if 

\end{definition} 
\section{Secret-sharing schemes}\label{appx:secretSharing}

We recall the standard definition of a secret-sharing scheme.

\begin{definition}[Secret sharing scheme \cite{Shamir:1979:SS:359168.359176}]
A \emph{$k$-out-of-$N$ secret sharing scheme} is a pair of algorithms $(\Share,\Recon)$ as follows.
$\Share$ takes as input a secret value $s$ and outputs a set of \emph{shares} $S=\{s_1,\dots,s_N\}$ such that
the following two properties hold.
\begin{itemize}
\item \emph{Correctness}: For any subset $S'\subseteq S$ of size $|S'|\geq k$, it holds that $\Recon(S')=s$, and
\item \emph{Privacy}: For any subset $S'\subseteq S$ of size $|S'|< k$, it holds that $H(s)=H(s|S')$, where $H$ denotes the binary entropy function.
\end{itemize}
$\Recon$ takes as input a (sub)set $S'$ of shares and outputs:
$$\Recon(S')=\begin{cases}\bot & \mbox{if} \qquad |S'|<k \\ s & \mbox{if} \qquad\exists S \mbox{ s.t. } S'\subseteq S \mbox{ and } \Share(s)=S \mbox{ and } |S'|\geq k \end{cases}.$$
\end{definition}
\ifproportionalfairness
	\input{APPX-mac}
\fi
\section{Constructions of time-line puzzles}\label{appx:timeLine}

\subsection{Black-box construction from inherently-sequential hash functions}\label{subsec:BBTimeLine}

\begin{definition}[Inherently-sequential hash function]
Let $\cH_\kappa=\{h_s:\{0,1\}^\kappa \rarr\{0,1\}^\kappa\}_{k\in\{0,1\}^n}$ for $n=\poly(\kappa)$ be a family of functions and 
suppose that evaluating $h_s(r)$ for $r\larr\{0,1\}^\kappa$ takes time $\Omega(T)$.
$\cH_\kappa$ is said to be \emph{inherently-sequential}
if evaluating $h_s^t(r)$ for $s\larr\{0,1\}^n,r\larr\{0,1\}^\kappa$ takes time $\Omega(t\cdot T)$,
and the output of $h_s^t(r)$ is pseudorandom.
\end{definition}

The time-line puzzle construction in this section relies on the following assumption about the existence of inherently-sequential functions.

\begin{assumption}\label{assume:seqHash}
There exists a family of functions $$\hatH_\kappa=\{\hath_s:\{0,1\}^\kappa \rarr\{0,1\}^\kappa\}_{s\in\{0,1\}^n}$$  which is inherently-sequential (where $n=\poly(\kappa)$).
\end{assumption}

\begin{definition}
$\SeqTimeLine$ is a time-line puzzle defined as follows, where $\hatH_\kappa$ is the inherently-sequential hash function family from Assumption~\ref{assume:seqHash}:
\begin{itemize}
\item $\Lock_m(1^\kappa,(d_1,\dots,d_m),(t_1,\dots,t_m))$ takes input data $(d_1,\dots,d_m)\in\{0,1\}^\kappa$, 
samples random values $s\larr\{0,1\}^n,x\larr\{0,1\}^\kappa$, and outputs the puzzle 
$$P=\left(x,(t_1,\dots,t_m),s,\left( d_1\oplus \hath_s^{t_1}(x),\dots,d_m\oplus \hath_s^{t_m}(x)\right)\right).$$
\item $\TimeStep_m(1^\kappa,i,x',a')$ outputs $\hath_{a'}(x')$.
\item $\Unlock_m(1^\kappa,x',b')$ outputs $x'\oplus b'$.
\end{itemize}
\end{definition}

It is clear that $\SeqTimeLine$ satisfies correctness, so we proceed to prove security.

\iffalse
\begin{theorem}
$\SeqTimeLine$ satisfies correctness.
\end{theorem}
\begin{proof}
easy...
\end{proof}
\fi

\begin{theorem}\label{thm:seqTimeLine}
If Assumption~\ref{assume:seqHash} holds, then $\SeqTimeLine$ is a secure time-line puzzle.
\end{theorem}
\begin{proof}
Given a time-line puzzle, in order to correctly output a piece of locked data $d_i$, the adversary $\Adv$ must compute the associated mask $\hath_s^{t_i}(x)$.
This is because 
\begin{itemize}
\item all components of the puzzle apart from the masked value $d_i\oplus \hath_s^{t_i}(x)$ are independent of the locked data $d_i$, and
\item the mask $\hath_s^{t_i}(x)$ is pseudorandom (by Assumption~\ref{assume:seqHash}), so the masked value $d_i\oplus \hath_s^{t_i}(x)$ is indistinguishable
from a truly random value without knowledge of the mask.
\end{itemize}
Moreover, by Assumption~\ref{assume:seqHash}, since $\hatH_\kappa$ is an inherently-sequential function family,
it holds that there is no (asympotically) more efficient way for a \PPT{} adversary to compute $\hath_s^{t_i}(x)$
than to sequentially compute $\hath_s$ for $t_i$ iterations. It follows that $\SeqTimeLine$ is a secure time-line puzzle.
\end{proof}

\subsection{Concrete construction based on modular exponentiation}\label{subsec:RSWTimeLine}

In this subsection we present an alternative construction quite similar in structure to  the above, but based on a concrete hardness assumption. 
Note that the \cite{RSW96} time-lock puzzle construction was also based on this hardness assumption, and our time-line puzzle may be viewed as a natural ``extension''
of their construction.

\begin{assumption}\label{assume:rsw}
Let $\RSA_\kappa$ be the distribution generated as follows: sample two $\kappa$-bit primes $p,q$ uniformly at random and output $N=pq$.
The family of functions $\cH^{square}=\{h_N:\ZZ_N\rarr\ZZ_N\}_{N\larr\RSA_\kappa}$, where the index $N$ is drawn from distribution $\RSA$
and $h_N(x)=x^2\mod N$, is inherently-sequential.
\end{assumption}

\begin{definition}\label{def:rswTimeLinePuzzle}
$\RSWTimeLine$ is a time-line puzzle defined as follows:
\begin{itemize}
\item $\Lock_m(1^\kappa,(d_1,\dots,d_m),(t_1,\dots,t_m))$ takes input data $(d_1,\dots,d_m)\in\{0,1\}^\kappa$,
samples random $\kappa$-bit primes $p,q$, sets $N=pq$, and outputs the puzzle
$$P=\left(x,(t_1,\dots,t_m),N,\left( d_1\oplus h_N^{t_1}(x),\dots,d_m\oplus h_N^{t_m}(x)\right)\right).$$
\item $\TimeStep_m(1^\kappa,i,x',a')$ outputs $h_{a'}(x')=x'^2\mod a'$.
\item $\Unlock_m(1^\kappa,x',b')$ outputs $x'\oplus b'$.
\end{itemize}
\end{definition}

Again, it is clear that $\RSWTimeLine$ satisfies correctness, so we proceed to prove security.

\begin{theorem}\label{thm:squareTimeLine}
If Assumption~\ref{assume:rsw} holds, $\RSWTimeLine$ is a secure time-line puzzle.
\end{theorem}
\begin{proof}
This follows from Assumption~\ref{assume:rsw} in exactly the same way as Theorem~\ref{thm:seqTimeLine} follows from Assumption~\ref{assume:seqHash},
so we refer the reader to the proof of Theorem~\ref{thm:seqTimeLine}.
\end{proof}

An advantage of this construction over $\SeqTimeLine$ is that the $\Lock$ algorithm can be much more efficient.
In the case of black-box inherently-sequential hash functions, 
we can only assume that the values $\hath_s^{t}(x)$ (which are XORed with the data values by the $\Lock$ algorithm)
are computed by sequentially evaluating $\hath_s$ for $t$ iterations -- that is, there is a linear dependence on $t$.
However, $\Lock$ can implemented much faster with the $\RSWTimeLine$ construction, as follows.
Since $p,q$ are generated by (and therefore, available to) the $\Lock$ algorithm, the $\Lock$ algorithm can efficiently compute $\phi(N)$. 
Then, $h_N^t(x)$ can be computed very efficiently by first computing $e=2^t\mod\varphi(N)$,
then computing $h_N^t(x)=x^e\mod N$. Exponentiation (say, by squaring) has only a logarithmic dependence on the security parameter.

Finally, we note that although both of the time-line puzzle constructions presented here lock $\kappa$ bits of data per puzzle 
(for security parameter $\kappa$), this is not at all a necessary restriction.
Using encryption, it is straightforwardly possible to lock much larger amounts of data for any given parameter sizes 
of the time-line puzzles presented here: for example, one can encrypt the data as $\Enc_k(d)$ using a secure secret-key encryption scheme, 
then use the given time-line puzzle schemes to lock the key $k$ (which is much smaller than $d$) under which the data is encrypted.
Such a scheme, with the additional encryption step, would be much more suitable for realistic use.

\end{appendix}

\end{document}